\documentclass[10pt,twocolumn,twoside]{IEEEtran}
\usepackage{graphicx}          %
\usepackage{amssymb}
\usepackage{amsmath}
\usepackage{calc}
\usepackage{tikz}
 \usepackage{pgfplots}
\usepackage{ushort}
\usepackage{xspace}
\usepackage{algorithm}
\usepackage{algpseudocode}
\usepackage{varwidth}
\usepackage{mathrsfs}
\usepackage{url}
\usetikzlibrary{arrows}
\usetikzlibrary{shapes}
 \usetikzlibrary{calc}
\usetikzlibrary{decorations.pathreplacing}
\usetikzlibrary{arrows,positioning} 
\usetikzlibrary{pgfplots.groupplots}

\def\undertilde#1{\mathord{\vtop{\ialign{##\crcr
$\hfil\displaystyle{#1}\hfil$\crcr\noalign{\kern1.5pt\nointerlineskip}
$\hfil\tilde{}\hfil$\crcr\noalign{\kern1.5pt}}}}}

\newcommand{\ub}[1]{\underbar{$#1$}}
\newcommand{\ol}[1]{\overline{#1}}
\newcommand{\ul}[1]{\underline{#1}}

\usepackage{array}

\usepackage{enumitem}
\usepackage{amsthm}
\usepackage{stmaryrd} %
\usepackage{setspace}
\usepackage{cite}
\newtheorem{definition}{Definition}
\newtheorem{thm}{Theorem}
\newtheorem{lemma}{Lemma}

\newtheorem{assum}{Assumption}

\newtheorem{corollary}{Corollary}
\theoremstyle{remark}
\newtheorem{rem}{Remark}

\renewcommand \Box      {\square}
\newcommand \Diam       {\lozenge}
\newcommand \Next       {\bigcirc}

\newcommand \Not        {\mathopen{\neg}}
\renewcommand \And      {\mathbin{\wedge}}
\newcommand \Or         {\mathbin{\vee}}
\renewcommand \Until      {\mathbin{\text{\sffamily U}\kern-.1em}}
\newcommand \Impl       {\mathbin{\rightarrow}}

\newcommand{\dens}{x}

\newcommand{\denscap}{\dens^\text{cap}}

\newcommand{\Verts}{\mathcal{V}}

\newcommand{\Links}{\mathcal{L}}

\newcommand{\Lin}{\mathcal{L}^\text{in}}

\newcommand{\Lout}{\mathcal{L}^\text{out}}

\newcommand{\flowin}{f^\text{in}}
\newcommand{\flowout}{f^\text{out}}

\newcommand{\head}{{\eta}}
\newcommand{\tail}{{\tau}}

\newcommand{\Domain}{\mathcal{X}}

\newcommand{\ltlphi}{\varphi}

\newcommand{\Dist}{\mathcal{D}}

\newcommand{\Post}{{\tt Post}}
\newcommand{\oPost}{{\overline{\Post}}}
\newcommand{\ux}{\underline{x}}
\newcommand{\barx}{\overline{x}}

\newcommand{\bx}{\boldsymbol{x}}
\newcommand{\by}{\boldsymbol{y}}

\newcommand{\bs}{\boldsymbol{s}}
\newcommand{\bd}{\boldsymbol{d}}
\newcommand{\bxi}{\boldsymbol{\xi}}
\newcommand{\bsigma}{\boldsymbol{\sigma}}
\newcommand{\xloc}{\bx^\text{loc}}

\newcommand{\sloc}{\bs^\text{loc}}
\newcommand{\xdown}{\bx^\text{down}}

\newcommand{\Ldown}{\Links^\text{down}}
\newcommand{\Lloc}{\Links^\text{loc}}
\newcommand{\Lup}{\Links^\text{up}}
\newcommand{\Ladj}{\Links^\text{adj}}

\newcommand{\sig}{s}
\newcommand{\bsig}{\boldsymbol{\sig}}
\newcommand{\Sig}{\mathcal{S}}
\newcommand{\sigloc}{\bsig^\text{loc}}

\newcommand{\T}{\mathcal{T}}

\newcommand{\Q}{\mathcal{Q}}
\newcommand{\QQ}{\mathbb{Q}}
\newcommand{\Z}{\mathcal{Z}}
\newcommand{\W}{\mathcal{W}}
\newcommand{\I}{\mathcal{I}}
\newcommand{\cl}{\mathbf{cl}}

\newcommand{\Taug}{\T_\text{aug}}
\newcommand{\toaug}{\to_\text{aug}}

\algrenewcommand\algorithmicindent{1.0em}
\newcommand{\pushcode}{\hspace{6em}\relax}
\newcommand{\pushcodeb}{\hspace{2.25em}\relax}

\definecolor{CornflowerBlue}{rgb}{0.258824,0.258824,0.435294}
\definecolor{cfblue}{rgb}{0.258824,0.258824,0.435294}
\definecolor{SkyBlue}{rgb}{0.196078,0.6,0.8}
\definecolor{dblue}{rgb}{.098,.243,.424}
\definecolor{lblue}{rgb}{.33,.57,.835}
\definecolor{llblue}{rgb}{.447,.643,.831}
\definecolor{lbluesam}{rgb}{.447,.643,.831} %
\definecolor{mblue}{rgb}{0.176, 0.380, 0.659}
\definecolor{lcomp}{rgb}{.969,.765,.416}
\definecolor{ddorange}{rgb}{0.624, 0.365, 0}
\definecolor{dorange}{rgb}{0.72, 0.506, 0.125}
\definecolor{lorange}{rgb}{0.961, 0.678, 0.165}
\definecolor{lgreen}{rgb}{.812,.969,.435}
\definecolor{dgreen}{RGB}{15,111,3}
\definecolor{lyellow}{rgb}{1,.859,.451}
\definecolor{dyellow}{rgb}{.651,.482,0}
\definecolor{lred}{rgb}{1,.6,.451}
\definecolor{dred}{rgb}{.65,.176,0}
\definecolor{dcompb}{RGB}{157,35,0}  %
\definecolor{lcompb}{RGB}{186,70,30}
\definecolor{llcompb}{RGB}{255,136,92}
\definecolor{lcompbsam}{RGB}{255,136,92}  %
\definecolor{dpurple}{RGB}{45,0,95}
\definecolor{mpurple}{RGB}{77,0,159}
\definecolor{lpurple}{RGB}{143,73,206}
\definecolor{purplea}{RGB}{122,24,207}

\makeatletter
\pgfdeclareshape{quad with diagonal fill}
{
    \inheritsavedanchors[from=rectangle]
    \inheritanchorborder[from=rectangle]
    \inheritanchor[from=rectangle]{north}
    \inheritanchor[from=rectangle]{north west}
    \inheritanchor[from=rectangle]{north east}
    \inheritanchor[from=rectangle]{center}
    \inheritanchor[from=rectangle]{west}
    \inheritanchor[from=rectangle]{east}
    \inheritanchor[from=rectangle]{mid}
    \inheritanchor[from=rectangle]{mid west}
    \inheritanchor[from=rectangle]{mid east}
    \inheritanchor[from=rectangle]{base}
    \inheritanchor[from=rectangle]{base west}
    \inheritanchor[from=rectangle]{base east}
    \inheritanchor[from=rectangle]{south}
    \inheritanchor[from=rectangle]{south west}
    \inheritanchor[from=rectangle]{south east}

    \inheritbackgroundpath[from=rectangle]
    \inheritbeforebackgroundpath[from=rectangle]
    \inheritbehindforegroundpath[from=rectangle]
    \inheritforegroundpath[from=rectangle]
    \inheritbeforeforegroundpath[from=rectangle]

    \behindbackgroundpath{%
        \pgfextractx{\pgf@xa}{\southwest}%
        \pgfextracty{\pgf@ya}{\southwest}%
        \pgfextractx{\pgf@xb}{\northeast}%
        \pgfextracty{\pgf@yb}{\northeast}%
            \def\pgf@diagonal@point@a{\southwest}%
            \def\pgf@diagonal@point@b{\northeast}%
          \pgfpathmoveto{\pgfpointorigin}%
        \pgfpathlineto{\northeast}%
        \pgfpathlineto{\pgfpoint{\pgf@xb}{\pgf@ya}}%
        \pgfpathclose
        \color{\pgf@diagonal@left@color}%
        \pgfusepath{fill}%

          \pgfpathmoveto{\pgfpointorigin}%
        \pgfpathlineto{\pgf@diagonal@point@a}%
        \pgfpathlineto{\pgfpoint{\pgf@xa}{\pgf@yb}}%
        \pgfpathclose
        \color{\pgf@diagonal@left@color}%
        \pgfusepath{fill}%

        \pgfpathmoveto{\pgfpoint{\pgf@xa}{\pgf@yb}}%
        \pgfpathlineto{\pgfpointorigin}%
        \pgfpathlineto{\pgf@diagonal@point@b}%
        \pgfpathclose
        \color{\pgf@diagonal@top@color}%
        \pgfusepath{fill}%

        \pgfpathmoveto{\pgfpoint{\pgf@xb}{\pgf@ya}}%
        \pgfpathlineto{\pgfpointorigin}%
        \pgfpathlineto{\pgf@diagonal@point@a}%
        \pgfpathclose
        \color{\pgf@diagonal@top@color}%
        \pgfusepath{fill}%
    }
}

\newif\ifpgf@diagonal@lefttoright
\def\pgf@diagonal@top@color{white}
\def\pgf@diagonal@left@color{gray!30}

\def\pgfsetnscolor#1{\def\pgf@diagonal@top@color{#1}}%
\def\pgfsetewcolor#1{\def\pgf@diagonal@left@color{#1}}%

\tikzoption{ns color}{\pgfsetnscolor{#1}}
\tikzoption{ew color}{\pgfsetewcolor{#1}}
\makeatother

\mathcode`l="8000
\begingroup
\makeatletter
\lccode`\~=`\l
\DeclareMathSymbol{\lsb@l}{\mathalpha}{letters}{`l}
\lowercase{\gdef~{\ifnum\the\mathgroup=\m@ne \ell \else \lsb@l \fi}}%
\endgroup

\title{Traffic Network Control from Temporal Logic Specifications}
\author{\thanks{This research was supported in part by the NSF under grants CNS-1446145 and CNS-1446151. Samuel Coogan and Murat Arcak are with the Department of Electrical Engineering and Computer Sciences, University of California, Berkeley, {\tt\{scoogan,arcak\}@eecs.berkeley.edu}. Ebru Aydin Gol is formerly with the Division of Systems Engineering, Boston University, {\tt  ebruaydin@gmail.com}. Calin Belta is with the Department of Mechanical Engineering, Boston University, { \tt cbelta@bu.edu} } Samuel Coogan, Ebru Aydin Gol, Murat Arcak, and Calin Belta }

\begin{document}
\maketitle

\begin{abstract}                          %
We propose a framework for generating a signal control policy for a traffic network of signalized intersections to accomplish control objectives expressible using linear temporal logic. By applying techniques from model checking and formal methods, we obtain a correct-by-construction controller that is guaranteed to satisfy complex specifications. To apply these tools, we identify and exploit structural properties particular to traffic networks that allow for efficient computation of a finite state abstraction. In particular, traffic networks exhibit a \emph{componentwise monotonicity} property which allows reach set computations that scale linearly with the dimension of the continuous state space. 
\end{abstract}

\section{Introduction}
\label{sec:introduction}

State-of-the-art approaches to coordinated control of signalized intersections often focus on limited objectives such as maximizing throughput \cite{Wongpiromsarn:2012yq} or maintaining stability of network queues \cite{Varaiya:2013xe, Varaiya:2013ty}; see \cite{Papageorgiou:2003ly} for a review of the literature. However, traffic networks are a natural domain for a much richer class of control objectives that are expressible using \emph{linear temporal logic (LTL)} \cite{clarke1999model, Baier:2008vn}. LTL \emph{formulae} allow control objectives such as ``actuate traffic flows such that throughput is always greater than $C_1$'' where $C_1$ is a threshold throughput, or such that ``traffic link queues are always less than $C_2$'' where $C_2$ is a threshold queue length. LTL formulae also allow more complex objectives such as ``infinitely often, the queue length on road $l$ should reach 0,'' ``anytime link $l$ becomes congested, it eventually becomes uncongested,'' or any combination of these conditions. As these examples suggest, many objectives that are difficult or impossible to address using standard control theoretic techniques are easily expressed in LTL.

In this paper, we propose a technique for synthesizing a signal control policy for a traffic network such that the network satisfies a given control objective expressed using LTL. The synthesized policy is a finite-memory, state feedback controller that is provably correct, that is, guaranteed to result in a closed loop system that satisfies the control objective.

Recent approaches to control synthesis from LTL specifications such as\cite{Tabuada:2006uq, Paulo:2008ys, Fainekos:2009rr, Kress-Gazit:2009mz, Kloetzer:2010ve, Abate:2011ul, Wongpiromsarn:2012oq, Yordanov:2012fk, Gol:2014nr, Julius:2012fk, Topcu:2012cs, Liu:2013gd, Aydin-Gol:2013kq,Plaku:2013vn, Coogan:2014pi} allow automatic development of correct-by-construction control laws; however, despite these promising developments, scalability concerns prevent direct application of existing results to large traffic networks.

To overcome these scalability limitations, we identify and exploit \emph{componentwise monotonicity} \cite{Kulenovic:2006kx} properties inherent in flow networks such as traffic networks. These properties allow efficient computation of bounds on the one-step reachable set from a rectangular box of initial conditions, which in turn allows efficient computation of a finite state abstraction of the dynamics, thereby mitigating a crucial bottleneck in the control synthesis process. A related approach to abstractions of monotone systems is suggested in \cite{Moor:2002fe}, however the componentwise monotonicity properties exploited in this work are much more general and encompass monotone systems as a special case. %
The present paper builds on our preliminary work in \cite{Coogan:2014bh} by defining componentwise monotonicity and identifying it as the enabling property for efficient abstraction. %

This paper is organized as follows: Section \ref{sec:preliminaries} gives necessary preliminaries. Section \ref{sec:sign-netw-traff} presents the model for signalized networks, and Section \ref{sec:ltl-spec-traff} establishes the problem formulation. Section \ref{sec:cw-monot-traff} identifies componentwise monotonicity properties of the traffic networks, and Section \ref{sec:finite-state-repr} presents scalable algorithms that rely on these properties to construct a finite state representation of the traffic network. Section \ref{sec:synth-contr-from} describes the controller synthesis approach, and discusses the computation requirements of our method. We present a case study in Section \ref{sec:case-study} and conclude our work in Section \ref{sec:conclusions}.

\section{Preliminaries}
\label{sec:preliminaries}

The set $\I\subseteq \mathbb{R}^n$ is a \emph{box} if it is the cartesian product of intervals, or equivalently, $\I$ is a box if there exists $x,y\in\mathbb{R}^n$ such that
  $\I=\prod_{i=1}^n\{z\in\mathbb{R}\mid x_i\prec_{i}^1 z \prec_{i}^2 y_i\}$
where $\prec_{i}^1,\prec_{i}^2\in\{<,\leq\}$ and $x_i, y_i$ denote the $i$th coordinate of $x$ and $y$, respectively. Defining $\prec^1\triangleq\{\prec^1_{i}\}_{i=1}^n$ and $\prec^2\triangleq\{\prec^2_{i}\}_{i=1}^n$, we may write $\I=\{z\in\mathbb{R}^n\mid x\prec^1 z\prec^2 y\}$. The vector $x$ is the \emph{lower corner} of $\I$, and likewise $y$ is the \emph{upper corner}.%

When applied to vectors, $<$, $\leq$, $>$, and $\geq$ are interpreted elementwise. The notation $\mathbf{0}$ denotes the all-zeros vector where the dimension is clear from context.
We denote closure of a set $Y$ by $\cl(Y)$. Given an index set $\Links$ and a set of values $x_l\in\mathbb{R}$ for $l\in\Links$, $\{x_l\}_{l\in\Links}$ denotes the collection of $x_l$, $l\in\Links$, but we also interpret $x=\{x_l\}_{l\in\Links}$ as an element of $\mathbb{R}^{|\Links|}$.

A \emph{transition system} is a tuple $\T=(Q,S,\to)$ where
$\Q$ is a finite set of states,
$\Sig$ is a finite set of actions, and
$\to\subset \Q\times \Sig\times \Q$ is a transition relation.
 We write $q\overset{s}{\to} q'$ instead of $(q,s,q')\in\to$. Note that all transition systems in this paper are \emph{finite} \cite{Baier:2008vn}. The evolution of a transition system is described by $\to$. That is, a transition system is initialized in some state $q_0\in\Q$, and, given an action $s\in\Sig$, the next state of the transition system is chosen nondeterministically from $\{q'\mid q\overset{s}{\to}q'\}$.

\section{Signalized Network Traffic Model}
\label{sec:sign-netw-traff}

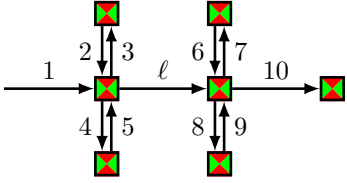
\begin{figure}
\centering
  \begin{tikzpicture}
    [subsystem/.style={quad with diagonal fill, ns color=red, ew color=green,draw,inner sep=1pt,minimum size=3mm},
>=to,line width=1pt, >=latex]
\node (Aleft) at (-1.5,0) {};
\node[subsystem] (A) at (0,0) {};
\node[subsystem] (B) at (1.5,0) {};
\node[subsystem] (C) at (0,1) {};
\node[subsystem] (D) at (1.5,1) {};
\node[subsystem] (E) at (0,-1) {};
\node[subsystem] (F) at (1.5,-1) {};
\node[subsystem] (G) at (3,0) {};
\draw[->] (Aleft) -- node[above]{$1$} (A);
\draw[->] (A) --node[above]{$l$} (B);
\draw[->] (B) -- node[above]{$10$}(G);
\draw[->] (A.70) -- node[right]{$3$}(C.290);
\draw[<-] (A.110) -- node[left]{$2$} (C.250);
\draw[->] (B.70) -- node[right]{$7$} (D.290);
\draw[<-] (B.110) -- node[left]{$6$} (D.250);
\draw[->] (E.70) -- node[right]{$5$}(A.290);
\draw[<-] (E.110) -- node[left]{$4$}(A.250);
\draw[->] (F.70) -- node[right]{$9$} (B.290);
\draw[<-] (F.110) -- node[left]{$8$} (B.250);
\end{tikzpicture}
  \caption{A typical traffic network with 11 links and 7 signalized intersections. In the figure, $\Ldown_l=\{l,7,8,10\}$, $\Lup_l=\{1,2,5\}$, and $\Ladj_l=\{3,4\}$. At each time step, a signal actuates a subset of upstream links.}
  \label{fig:sig}
\end{figure}

A signalized traffic network consists of a set $\Links$ of \emph{links} and a set $\Verts$ of \emph{signalized intersections}. For $l\in\Links$, let $\head(l)\in\Verts$ denote the downstream intersection of link $l$ and let $\tail(l)\in\Verts\cup \emptyset$ denote the upstream intersection of link $l$. A link $l$ with $\tail(l)=\emptyset$ serves as an entry-point into the network, and we assume $\head(l)\neq \tail(l)$ for all $l\in\Links$ (\emph{i.e.}, no self-loops). Link $k\neq l$ is \emph{upstream} of link $l$ if $\head(k)=\tail(l)$, \emph{downstream} of link $l$ if $\tail(k)=\head(l)$, and \emph{adjacent to} link $l$ if $\tail(k)=\tail(l)$.  Roads exiting the traffic network are not modeled explicitly. For each $v\in\Verts$, define
$  \Lin_v=\{l\mid\head(l)=v\}$, $\Lout_v=\{l\mid\tail(l)=v\}$
 and for each $l\in\Links$, define     %
 \begin{align}
   \label{eq:8}
     \Lup_l&=\{k\in\Links\mid \head(k)=\tail(l)\}\\
     \Ldown_l&=\{k\in\Links\mid \tail(k)=\head(l)\}\cup \{l\} \\
  \Ladj_l&=\{k\in\Links\mid \tail(k)=\tail(l)\}\backslash \{l\}
 \end{align}
so that $\Ldown_l$ includes link $l$ and the links downstream of link $l$, and $\Lup_l$ and $\Ladj_l$ are the links upstream and adjacent to $l$, respectively, see Fig. \ref{fig:sig}.
 We have $\Ldown_l\cap\Ladj_l=\emptyset$ and $\Lup_l\cap\Ladj_l=\emptyset$, but note that it is possible for $\Ldown_l\cap\Lup_l\neq \emptyset$, in particular,  if there is a cycle of length two in the network. Let $\Lloc_l=\Ldown_l\cup \Lup_l\cup \Ladj_l$ be links ``local'' to link $l$.

 Each link $l\in \Links$ possesses a queue 
$\dens_l[t]\in[0,\denscap_l]$
representing the number of vehicles on link $l$ at time step $t\in\mathbb{N}\triangleq \{0,1,2,\ldots\}$ where $\denscap_l$ is the capacity of link $l$. We allow $\dens_l$ to be a continuous quantity, thus adopting a fluid-like model of traffic flow evolving in slotted time as in \cite{Wongpiromsarn:2012yq, Varaiya:2013ty, Varaiya:2013xe}.

Movement of vehicles among link queues is governed by mass-conservation laws and the state of the signalized intersections. A link is said to be \emph{actuated} if outgoing flow from link $l$ is allowed as determined by the state of the traffic signal at intersection $\head(l)$. At each intersection $v$,
\begin{align}
  \label{eq:1}
  \Sig_v\subseteq 2^{\Lin_v}
\end{align}
denotes the set of available signal \emph{phases}, that is, each $\sig_v\in \Sig_v$, $\sig_v\subseteq {\Lin_v}$ denotes a set of incoming links at intersection $v$ that may be actuated simultaneously. We define 
\begin{align}
  \label{eq:9}
  \Sig=\{\cup_{v\in\Verts}s_v\mid s_v\in\Sig_v\ \forall v\in\Verts\}\subseteq 2^\Links
\end{align}
so that each $\bsig\in\Sig$, $\bsig\subseteq \Links$ denotes a set of links in the network that may be actuated simultaneously. We identify $\bsig\in\Sig$ with its constituent phases so that $\bsig=\{s_v\}_{v\in\Verts}$, and we interpret $\Sig$ as the set of allowed inputs to the traffic network. 

When a link is actuated, a maximum of $c_l$ vehicles are allowed to flow from link $l$ to links $\Lout_{\head(l)}$ per time step where $c_l$ is the known \emph{saturation flow} for link $l$, \cite{Papageorgiou:2003ly}. The \emph{turn ratio} $\beta_{lk}$ denotes the fraction of vehicles exiting link $l$ that are routed to link $k$, \cite{Varaiya:2013xe}. Then $\beta_{lk}\neq 0$ only if $\head(l)=\tail(k)$, and
\begin{align}
  \label{eq:3}
  \sum_{k\in\Lout_{\head(l)}}\beta_{lk}\leq 1.
\end{align}
Strict inequality in \eqref{eq:3} implies that a fraction of vehicles on link $l$ are routed off the network via unmodeled roads that exit the network. Traffic flow can occur only if there is available capacity downstream. To this end, the supply ratio $\alpha^{\sig_v}_{lk}$ denotes the fraction of link $k$'s capacity available to link $l$ during phase $\sig_v\in\Sig_{\tail(k)}$. That is, link $l$ may only send $\alpha^{\sig_v}_{lk}(\denscap_k-\dens_k[t])$ vehicles to link $k$ in time period $t$ under input $\sig_v$. As the supply is only divided among actuated incoming links, it follows that for each $k\in\Links$
\begin{align}
  \label{eq:4}
\sum_{l\in\sig_{v}}\alpha^{\sig_v}_{lk}=1\quad \forall \sig_{v}\in\Sig_{\tail(k)}, \sig_v\neq \emptyset.
\end{align}
Constant turn and supply ratios are a common modeling assumption justified by empirical observations; see \cite{Lebacque:2005fk} for further discussion. 

We are now in a position to define the dynamics of the link queues. As we will see subsequently, the flow of vehicles out of link $l$ is only a function of the state of links in $\Ldown_l$, and the update of link $l$'s state is only a function of links in $\Lloc_l$.

Let $\bx[t]=\{x_l[t]\}_{l\in\Links}$, $\xdown_l[t]=\{x_k[t]\}_{k\in\Ldown_l}$, and $\xloc_l[t]=\{x_k[t]\}_{k\in\Lloc_l}$. {The} outflow of link $l\in\Links$ {is }as follows:
\begin{align}
\notag &\flowout_l(\xdown_l,\sig_{\head(l)})=\\
  \label{eq:5}&\begin{cases}
\min\bigg\{x_l[t],c_l,\min_{\substack{k \text{ s.t.}\\\beta_{lk}\neq0}}\Big\{\frac{\alpha^{\sig_{\head(l)}}_{lk}}{\beta_{lk}}(\denscap_k-\dens_k[t])\Big\}\bigg\}&\\
\hspace{2.15in}\text{if $l\in\sig_{\head(l)}$}&\\
0 \hspace{2.4in}\text{else}.
\end{cases}
\end{align}
{The interpretation of \eqref{eq:5} is that the flow of vehicles exiting a link $l$ when actuated is the minimum of the link's queue length, its saturation flow, and the downstream \emph{supply} of capacity, weighted appropriately by turn and supply ratios. This modeling approach is based on the \emph{cell transmission model} of traffic flow \cite{Daganzo:1994fk} which restricts flow if there is inadequate capacity downstream. A consequence of \eqref{eq:5} is that inadequate capacity on one downstream link at an intersection causes congestion that blocks incoming flow to other downstream links. %
This phenomenon, sometimes called the \emph{first-in-first-out} property, has been widely studied in the transportation literature and occurs even in multilane settings \cite{Munoz:2002qv}\footnote{{Even if a turn pocket exists at an intersection, it is often too short to fully mitigate this blocking property. Nonetheless, if the road geometry is such that a sufficient number of dedicated lanes exist for a turning movement, these lanes may be modeled with a separate link.}}.}
The number of vehicles in each link's queue then evolves according to the mass conservation equation
\begin{align}
  \label{eq:7}
  x_l[t+1]=&F_l(\xloc_l[t],\sigloc_l[t],d_l[t])\\
\notag\triangleq &\min\Big\{\denscap_l ,x_l[t]-\flowout_l(\xdown_l[t],\sig_{\head(l)})\\
\label{eq:7-2}&\ \ +\sum_{j\in\Lup_l}\beta_{jl}\flowout_j(\xdown_j[t],\sig_{\head(j)})+d_l[t]\Big\}
\end{align}
where $d_l[t]$ is the number of vehicles that exogenously enters the queue on link $l$ in time step $t$, $\bd=\{d_l[t]\}_{l\in\Links}$, and $\sigloc_l=\{\sig_{\head(l)},\sig_{\tail(l)}\}$ if $\tail(l)\neq \emptyset$, $\sigloc_l=\{\sig_{\head(l)}\}$ otherwise; that is,  $\sigloc_l$ is the state of the signals that are ``local'' to link $l$. The minimization in \eqref{eq:7-2} is only needed in case the exogenous input $d_l[t]$ would cause the state of link $l$ to exceed $\denscap_l$ and ensures that the network dynamics maps
\begin{align}
  \label{eq:15}
  \Domain=\prod_{l\in\Links}[0,\denscap_l]
\end{align}
to itself. We interpet this as refusal of vehicles attempting to exogenously enter the network when the link is full. Note in particular that the supply/demand formulation prevents upstream inflow from exceeding supply and thus for links with no exogenous input, $\denscap_l$ is never the unique minimizer in \eqref{eq:7-2}.

\begin{rem}
  An alternative to the above approach is to define an auxiliary sink state $\texttt{Out}$ in the transition systems $\T$ defined in Section \ref{sec:finite-state-repr} which captures any trajectories that exit the domain $\Domain$. The temporal logic specification can then incorporate the requirement that the system never enters this $\texttt{Out}$ state.  %
\end{rem}

\begin{assum}
 We assume there exists $\mathcal{D}\subset \mathbb{R}^{\Links}$ such that
\begin{align}
  \label{eq:12}
\bd[t]\in\mathcal{D}  \quad \forall t
\end{align} 
and $\Dist$ satisfies $\Dist\subset\cup_{i=1}^{n_\Dist}\Dist_i$ where each $\Dist_i$ is given by
\begin{align}
  \label{eq:27}
  \Dist_i=\{\bd\mid \ul{\bd}^i\leq \bd \leq \ol{\bd}^i\}%
\end{align}
for some $\ul{\bd}^i=\{\ul{d}_l^i\}_{l\in\Links}$, $\ol{\bd}^i=\{\ol{d}_l^i\}_{l\in\Links}$.
\end{assum}
In other words, we assume the disturbance is contained within a union of boxes given by \eqref{eq:27}. This assumption is not particularly restrictive, as any compact subset of $\mathbb{R}^{\Links}$ can be approximated with boxes to arbitrary precision \cite{Kieffer:2002qf}, however the number of boxes $n_\Dist$ affects the computation time as detailed in Section \ref{sec:comp-requ}.
 
We let $F(\bx,\bs,\bd)=\{F_l(\xloc_l,\sloc_l,d_l)\}_{l\in\Links}:\Domain\times \Sig \times \Dist\to \Domain$ so that 
\begin{align}
  \label{eq:11}
\bx[t+1]=F(\bx[t],\bs[t],\bd[t]).
\end{align}

The set of states of system \eqref{eq:11} that are reachable from a set $Y\subset \Domain$ under the control signal $\bs\in\Sig$ in one timestep is denoted by the $\Post$ operator and given by
\begin{align}
  \label{eq:29}
  \Post(Y,\bsig)=\{\bx'=F(\bx,\bsig,\bd)\mid \bx\in Y, \bd\in\Dist\}.  
\end{align}
We call $\Post(Y,\bsig)$ the \emph{one step reachable set from $Y$ under $\bsig$}. The main features of the queue-based modeling approach proposed above such as finite saturation rates, finite queue capacity, a set of available signaling phases, and fixed turn ratios are standard in many modeling and simulation approaches such as \cite{Varaiya:2013ty, Varaiya:2013xe}, see also \cite{Papageorgiou:1995rm, Papageorgiou:2003ly} and references therein for discussions of queue-based modeling of traffic networks.

\section{Problem Formulation and Approach}
\label{sec:ltl-spec-traff}

We now define and motivate the need for control objectives expressible in LTL for traffic networks, and we outline a control synthesis approach which relies on a finite state representation of the traffic dynamics to meet these objectives.

LTL \emph{formulae} are generated inductively using the Boolean operators $\Or$ (disjunction), $\And$ (conjunction), $\Not$ (negation), and the temporal operators $\Next$ (next) and $\Until$ (until). From these, we obtain a suite of derived logical and temporal operators such as $\Impl$ (implication), $\Box$ (always), $\Diam$ (eventually), $\Box\Diam$ (infinitely often), {finite deadlines with repeated $\Next$,} and many others, see \cite{Baier:2008vn, clarke1999model}.

Formally, such formulae are expressed over a set of atomic propositions, which we restrict to be indicator expressions over subsets of $\Domain$ or predicates over the signaling state. For example, the atomic proposition $x_l\leq 10$ is true for all $\bx\in\Domain$ that satisfies the condition $x_l\leq 10$ (which constitutes a box subset of $\Domain$), and the atomic proposition $l\in\bsig$ is true for all signals that actuate link $l$. We will see in Section \ref{sec:cw-monot-traff} and Section \ref{sec:finite-state-repr} that restricting to atomic propositions corresponding to box subsets of $\Domain$ offers significant computational advantages.

  Semantically, LTL formulae are interpreted over a trajectory $\bx[t]$ and the corresponding input sequence $\bsig[t]$ for $t=0,1,\ldots$. For example, the state/input sequence $(\bx[t],\bsig[t])$ \emph{satisfies} the LTL formula 
$\ltlphi=\Box(x_l\leq 10)\And \Box\Diam(l\in \bsig)    $
if and only if $x_l[t]\leq 10$ for all $t$ and $l\in\bsig[t]$ infinitely often (\emph{i.e.}, for infinitely many $t$). Thus a trajectory satisfies a LTL formula if and only if the formula holds for the corresponding trace of atomic propositions that are valid at each time step.  A formal definition of the semantics of LTL over traces is readily available in the literature, \emph{e.g.},  \cite{Baier:2008vn, clarke1999model}, and is a natural interpretation of the above Boolean and temporal operators. For example, a trace satisfies $\Diam\ltlphi$ if and only if there exists a suffix of the trace satisfying $\ltlphi$.

Examples of LTL formulae representing desired control objectives relevant to traffic networks include those from the Introduction, as well as:
\begin{itemize}[leftmargin=10pt]
\item $\ltlphi_1=\Diam \Box (x_l\leq C)$ for some $C$\\
``Eventually, link $l$ will have less than $C$ vehicles and this will remain true for all time''
\item $\ltlphi_2=\Box \Diam (l\in\bsig)$\\
``Infinitely often, link $l$ is actuated''
\item $\ltlphi_3=\Box((l\in\sig_{v_1}) \Impl \Next (k\in\sig_{v_2}))$\\
``Whenever signal $v_1$ actuates link $l$, signal $v_2$ must actuate link $k$ in the next time step''
\item $\ltlphi_4=\Box(x_{l}\geq C_1 \Impl \Diam(x_l\leq C_2))$\\
``Whenever the number of vehicles on link $l$ exceeds $C_1$, it is eventually the case that the number of vehicles on link $l$ decreases below $C_2$.''
\end{itemize}
The main problem considered in this paper is as follows:

\vspace{5pt}
\noindent\textbf{Control Synthesis Problem.} \emph{Given a traffic network and an LTL formula $\ltlphi$ over a set of atomic propositions as described above, find {a control strategy that, at each time step, chooses a signaling input} such that all trajectories of the traffic network satisfy $\ltlphi$ from any initial condition.}

To solve the control synthesis problem, we propose computing a finite state abstraction that simulates (in a manner to be formalized below) the traffic network dynamics. {As we discuss in Section \ref{sec:synthesis-summary}, the result is a full-state feedback controller which requires finite memory.}  We rely on dynamical properties of the traffic network to compute the abstraction, and then apply tools from automata theory and formal methods to synthesize a finite-memory, state feedback control strategy solving the control synthesis problem.

\section{Componentwise Monotonicity of Traffic Networks}
\label{sec:cw-monot-traff}
To generate control strategies for the traffic network that guarantee satisfaction of a LTL formula, we first construct a finite state representation, or \emph{abstraction}, of the model defined in Section \ref{sec:sign-netw-traff}. We now define a \emph{componentwise monotonicity} property that simplifies this task and next show that the model in Section \ref{sec:sign-netw-traff} possesses this property.

\begin{definition}
\label{def:mm}
Consider the dynamical system 
\begin{equation}
  \label{eq:16}
z[t+1]=f(z[t],w[t])
\end{equation}
for $z\in\Z\subseteq\mathbb{R}^{n}$, $w\in\W\subseteq\mathbb{R}^m$ with $f:\Z\times\W\to\Z$ continuous. System \eqref{eq:16} is \emph{componentwise monotone} if there exists a \emph{signature matrix} $\Delta=[\delta_{ij}]_{i,j=1}^n$ with each $\delta_{ij}\in\{-1,1\}$ such that for all $i$,
\begin{align}
  \label{eq:47}
\delta_{ij}\ul{\xi}_j\leq \delta_{ij}\ol{\xi}_j \text{ and }\ul{w}_j\leq\ol{w}_j\ \forall j\in\Links\\
\text{implies}\qquad f(\ul{\xi},\ul{w})\leq f(\ol{\xi},\ol{w}))
\end{align}
for any $\ul{\xi},\ol{\xi}\in\Z$,  $\ul{w},\ol{w}\in \W$.
That is, \eqref{eq:16} is componentwise monotone if $f_i$ is monotonic in each $z$ variable and monotonically increasing in each $w$ variable.
\end{definition}

A definition similar to Definition \ref{def:mm} appears in \cite{Kulenovic:2006kx}, but omits dependence on a disturbance input.  %

 We now give a characterization of componentwise monotone systems which stipulates that $\partial f/\partial z$ be \emph{sign-stable}, that is, the sign structure of the Jacobian does not change as $z,w$ range over their domain.

\begin{lemma}
\label{lem:main}
  Consider dynamical system \eqref{eq:16} and further suppose that 
$f(z,w)=\begin{bmatrix}f_1(z,w)&\ldots f_n(z,w)\end{bmatrix}^T    $
 is Lipschitz continuous so that partial derivatives exist almost everywhere. If for all $i\in\{1,\ldots,n\}$: %
  \begin{align}
    \label{eq:17}
   \hspace{-3pt}\forall j\in\{1,\ldots,n\}\ \exists \delta_{ij}\in\{-1,1\}&:\delta_{ij}\frac{\partial f_i}{\partial z_j}(z,w)\geq 0 \emph{ a.e. }\hspace{-4pt}\\
    \label{eq:18}
\text{and }\qquad\qquad    \forall j\in\{1,\ldots,m\}&:\frac{\partial f_i}{\partial w_j}(z,w)\geq 0 \emph{ a.e. }
  \end{align}
where a.e.(almost everywhere) implies the condition must hold wherever the derivative exists, then \eqref{eq:16} is componentwise monotone.
\end{lemma}

\begin{proof}
 Let $[\delta_{ij}]_{i,j=1}^n$ be as in the hypothesis of the Lemma.
By the Fundamental Theorem of Calculus, for all $\ul{w}$ and for almost all\footnote{Eq. \eqref{eq:13} requires existence of $\partial g/\partial z$ almost everywhere along the line segment connecting $\ul{z}$ and $\ol{z}$, which holds for almost all $\ol{z}$ for fixed $\ul{z}$ \cite[Ch. 2]{Clarke:1990kx}. Similarly, $f_i(\ol{\xi},\ol{w})-f_i(\ol{\xi},\ul{w})$ holds for almost all $\ol{w}$ for fixed $\ul{w}$.} $\ul{\xi}$, $\ol{\xi}$ satisfying $\delta_{ij}\ul{\xi}_j\leq \delta_{ij} \ol{\xi}_j$ for all $j$,%
\begin{align}
  \label{eq:13}
&    f_i(\ol{\xi},\ul{w})-f _i(\ul{\xi},\ul{w})=\\
&\textstyle \left(\int_{0}^1\sum_{j=1}^n\frac{\partial f_i}{\partial \xi_j}(\ul{\xi}+r(\ol{\xi}-\ul{\xi}),\ul{w})(\ol{\xi}_j-\ul{\xi}_j)dr\right)\geq 0
\end{align}
where nonnegativity follows because $\delta_{ij}(\ol{\xi}_j-\ul{\xi}_j)\geq 0$, $\delta_{ij}\partial f_i/\partial \xi_j\geq 0$, and $\delta_{ij}^2=1$ for all $i,j$. Similarly, for almost all $\ol{w}\geq\ul{w}$, $  f_i(\ol{\xi},\ol{w})-f_i(\ol{\xi},\ul{w})\geq 0$.
It follows by continuity of $f_i$ that $f_i(\ol{\xi},\ol{w})-f_i(\ul{\xi},\ul{w})\geq 0$ for all $\ul{\xi}$, $\ol{\xi}$ that satisfy $\delta_{ij}\ul{\xi}_j\leq \delta_{ij}\ol{\xi}_j$ and all $\ul{w}\leq\ol{w}$, for all $i$, completing the proof.
\end{proof}

The critical feature of  componentwise monotone systems we wish to exploit is that over approximating the one-step reachable set from a box of initial conditions is computationally efficient. In particular, the reach set is contained within a box defined by the value of $f_i$ at two particular points for each $i$, regardless of the dimension of the spaces $\Z$ and $\W$:

\begin{lemma}
  \label{lem:reach}
Let \eqref{eq:16} be componentwise monotone with signature matrix $\Delta=[\delta_{ij}]_{i,j=1}^n$ and assume $\Z$ is a closed box. Given $\ul{z},\ol{z}\in\Z$ and $\ul{w},\ol{w}\in\W$ with $\ul{z}\leq\ol{z}$ and $\ul{w}\leq\ol{w}$. Let $\ul{\xi}^i\in\Z$ and $\ol{\xi}^i\in\Z$ be defined elementwise as follows for each $i$:
\begin{alignat}{2}
  \label{eq:48}
\ul{\xi}^i_j&=
\begin{cases}
  \ul{z}_j&\text{if }\delta_{ij}=1\\
  \ol{z}_j&\text{if }\delta_{ij}=-1
\end{cases}, \quad \ol{\xi}^i_j&=
\begin{cases}
  \ol{z}_j&\text{if }\delta_{ij}=1\\
  \ul{z}_j&\text{if }\delta_{ij}=-1.
\end{cases}
\end{alignat}

Then
\begin{align}
  \label{eq:22}
  f_i(\ul{\xi}^i,\ul{w})\leq f_i(z,w)\leq f_i(\ol{\xi}^i,\ol{w})\qquad \forall i
\end{align}
for all $z,w$ such that $\ul{z}\leq z \leq \ol{z}$ and $\ul{w}\leq w \leq \ol{w}$.
\end{lemma}
\begin{proof}
  Observe that $\delta_{ij}\ul{\xi}^i_j\leq \delta_{ij}z_j$ for all $i,j$ for all $\ul{z}\leq z\leq \ol{z}$, and symmetrically, $\delta_{ij}z_j\leq\delta_{ij}\ol{\xi}^i_j$ for all $i,j$ for all $\ul{z}\leq z\leq \ol{z}$. The Lemma then follows immediately from Definition \ref{def:mm}.
\end{proof}
This remarkable feature of componentwise monotone systems is analogous to well-known results for monotone systems \cite{Hirsch:1985fk, Angeli:2003fv, Smith:2008fk}, but componentwise monotonicity allows consideration of a much broader class of systems, including traffic networks, which are generally not monotone. %

\begin{rem}
  \label{rem:tight}
  Observe that the lower and upper bounds in \eqref{eq:22} are achieved for appropriate choice of $z$ and $w$, thus the approximation of the one-step reachable set is tight.
\end{rem}

To prove that the traffic network dynamics developed in Section \ref{sec:sign-netw-traff} are componentwise monotone, we first require a technical assumption:
\begin{assum}
\label{assum:1}
  For all $l\in\Links$,
  \begin{align}
    \label{eq:38}
    c_l\leq \denscap_l -\frac{\beta_{kl}}{\alpha_{kl}}c_k\quad \forall k\in\Lup_l.
  \end{align}
\end{assum}
Assumption \ref{assum:1} is a sufficient condition for ensuring that {if a link has inadequate capacity and blocks upstream flow, then this link's queue will not empty in one time step. This effectively is an assumption that the time step is sufficiently small to appropriately capture the queuing phenomenon. Specifically, the saturation} flow rate $c_l$ of link $l$ is in units of vehicles per time step and, thus, is implicitly a function of the chosen time step. Physically, $c_l$ is required to decrease with decreased time step and thus Assumption \ref{assum:1} is satisfied when a sufficiently small time step is used for the model.

\begin{thm}
  \label{thm:mm_traffic}
The traffic network model is componentwise monotone for any signaling input $\bs\in\Sig$. In particular, $F_l$ is increasing in $x_k$ for $k$ downstream or upstream of link $l$ or equal to $l$, and decreasing in $x_k$ for $k$ adjacent to link $l$. 
\end{thm}
\begin{proof}
For  fixed $\bs\in\Sig$, we show that $F(\bx,\bs,\bd)$ satisfies conditions \eqref{eq:17} and \eqref{eq:18} of Lemma \ref{lem:main} with $\bx,\bd$ replacing $z,w$.  Observe that $F$ is continuous and piecewise differentiable by \eqref{eq:5}--\eqref{eq:7-2} and \eqref{eq:11}, thus it is Lipschitz continuous \cite{Scholtes:2012fk}. {The minimum function in \eqref{eq:5} implies that $F$ is differentiable almost everywhere.} We first have
 $ \frac{\partial F_l}{\partial d_l}\in\{0,1\}$ \text{a.e.}
by \eqref{eq:7-2}, satisfying \eqref{eq:18}. Now consider $\partial F_l/\partial x_k$. For \eqref{eq:18}, we consider four exhaustive cases:

  \begin{itemize}[leftmargin=10pt]
\item Case 1, $k\in(\Ldown_l\cup \Lup_l)\backslash \{l\}$.  From \eqref{eq:5}--\eqref{eq:7-2}{, link $k$ may block the outflow of link $l$ when $k\in\Ldown_l$, or link $k$ may contribute to the inflow to link $l$ if $k\in\Lup_l$, thus} we have $\frac{\partial F_l}{\partial x_k}\in\{0,-\frac{\partial \flowout_l}{\partial x_k},\beta_{kl}\frac{\partial \flowout_k}{\partial x_k},-\frac{\partial \flowout_l}{\partial x_k}+\beta_{kl}\frac{\partial \flowout_k}{\partial x_k}\}$ a.e. where the fourth possibility occurs only if $k\in\Ldown_l\cap\Lup_l$. But $\frac{\partial \flowout_l}{\partial x_k}\in \{0,-{\alpha_{lk}^{\sig_{\head(l)}}}/{\beta_{lk}}\}$ a.e. and $\frac{\partial \flowout_k}{\partial x_k}\in\{0,1\}$ a.e., thus $\frac{\partial F_l}{\partial x_k}\geq 0$ a.e., satisfying \eqref{eq:17}.

\item Case 2, $k=l$. We have $\frac{\partial \flowout_l}{\partial x_l}\in\{0,1\}$ a.e. and, for $j\in\Lup_l$, $\frac{\partial \flowout_j}{\partial x_l}\in\{0,\alpha_{jl}^{\sig_{\head(j)}}/\beta_{jl}\}$ a.e., however, Assumption \ref{assum:1} ensures that, a.e., either  $\frac{\partial \flowout_l}{\partial x_l}=0$ or {$\frac{\partial \flowin_l}{\partial x_l}=0$, \emph{i.e.}, }$\frac{\partial \flowout_j}{\partial x_l}=0$ for all $j\in\Lup_l$. Thus $\frac{\partial F_l}{\partial x_l}\in\{0,1,1+\sum_{j\in \Lup_l}\beta_{jl}\frac{\partial \flowout_j}{\partial x_l}\}$ a.e. But $\sum_{j\in \Lup_l}\beta_{jl}\frac{\partial \flowout_j}{\partial x_l}\geq -\sum_{j\in\Lup_l}\alpha^{\sig_{\head(j)}}_{jl}= -1$ by \eqref{eq:4} (recall that $\head(j)=\tail(l)$ for all $j\in\Lup_l$), {that is, $\partial \flowin_l/\partial x_l\geq -1$}, thus $\frac{\partial F_l}{\partial x_l}\geq 0$ a.e., satisfying \eqref{eq:17}.
\item Case 3, $k\in\Ladj_l$. {In this case, inadequate capacity of link $k$ may block flow to link $l$, as discussed above.} We have $\frac{\partial F_l}{\partial x_k}=\sum_{j\in\Lup_l}\beta_{jl}\frac{\partial \flowout_j}{\partial x_k}$. Since $\frac{\partial \flowout_j}{\partial x_k}\in \{0,-{\alpha_{jk}^{\sig_{\head(j)}}}/{\beta_{jk}}\}$ a.e., we have $\frac{\partial F_l}{\partial x_k}\leq 0$ a.e., satisfying \eqref{eq:17}.
\item Case 4, $k\not\in \Lloc_l$. Then $\frac{\partial F_l}{\partial x_k}=0$, trivially satisfying \eqref{eq:17}.
\end{itemize}

\end{proof}
The following corollary implies that the one-step reachable set of the traffic dynamics from a (closed) box  $\I$ for any given signaling input $\bsig$ is over-approximated by the union of boxes, one box for each $i=1,\dots n_\Dist$, where each of these boxes is efficiently computed by evaluating $F_l$ at two particular points for each $l\in\Links$. The obtained over-approximation is denoted with the $\oPost$ operator. This critical result allows efficient computation of a finite state representation of the traffic dynamics, as detailed in Section \ref{sec:finite-state-repr}. 

\begin{corollary}
\label{cor:1}
Consider the set $\I=\{\bx\mid \ul{\bx}  \leq \bx \leq \bar{\bx}\}$ for $\ul{\bx},\ol{\bx}\in\Domain$, and for each $l\in\Links$, define $\ul{\bxi}^l(\ul{\bx},\ol{\bx})=\{\ul{\xi}^l_k(\ul{x}_k,\ol{x}_k)\}_{k\in \Lloc_l}$, $\ol{\bxi}^l(\ul{\bx},\ol{\bx})=\{\ol{\xi}^l_k(\ul{x}_k,\ol{x}_k)\}_{k\in \Lloc_l}$ where
\begin{align}
  \label{eq:28}
     \ul{\xi}^l_k(\ul{x}_k,\ol{x}_k)&=
     \begin{cases}
       \ul{x}_k&\text{if $k\in\Ldown_l\cup \Lup_l$}\\
       \ol{x}_k&\text{if $k\in\Ladj_l$}
     \end{cases}\\
  \label{eq:28-2}     \ol{\xi}^l_k(\ul{x}_k,\ol{x}_k)&=
     \begin{cases}
       \ol{x}_k&\text{if $k\in\Ldown_l\cup \Lup_l$}\\
       \ul{x}_k&\text{if $k\in\Ladj_l$}.
     \end{cases}
   \end{align}
Then for all $\bsig\in\Sig$, $\Post(\I,\bsig)\subseteq \oPost(\I,\bsig)$ where
\begin{align}
\notag &\oPost(\I,\bsig):=\\
  \label{eq:30}&\bigcup_{i=1}^{n_\Dist}\{\bx'\mid F_l(\ul{\bxi}^l,\sloc_l,\ul{d}^i_l)\leq x'_l \leq F_l(\ol{\bxi}^l,\sloc_l,\ol{d}^i_l)\quad \forall l\in\Links\}.
\end{align}
\end{corollary}
\begin{proof}
  By substituting $\ul{\bx},\ol{\bx}$ for $\ul{z},\ol{z}$ and $\ul{d}_l^i,\ol{d}_l^i$ for $\ul{w},\ol{w}$ in Lemma \ref{lem:reach} and defining $f(\bx,\bd)\triangleq F(\bx,\bsig,\bd)$, we obtain $\{\bx'=F(\bx,\bsig,\bd)\mid \bx\in\I,\bd\in\Dist_i\}\subseteq \{\bx'\mid F_l(\ul{\bxi}^l,\sloc_l,\ul{d}^i_l)\leq x'_l \leq F_l(\ol{\bxi}^l,\sloc_l,\ol{d}^i_l)\ \forall l\in\Links\}$ for all $i=1,\ldots,n_\Dist$.
The corollary follows from the trivial fact that $\Post(\I,\bsig)=\cup_{i=1}^{n_\Dist}\{\bx'=F(\bx,\bsig,\bd)\mid \bx\in\I,d\in\Dist_i\}$.
\end{proof}

\section{Finite State Representation}
\label{sec:finite-state-repr}

To apply the powerful tools of LTL synthesis, we require a finite state representation of the traffic network model.  In general, obtaining finite state abstractions is a difficult problem and existing techniques do not scale well. In this section, we exploit the componentwise monotonicity properties developed above and propose an efficient method for determining a finite state representation of the traffic network dynamics.
\subsection{Finite State Abstraction}

\begin{definition}[Box partition]
  For finite index set $\Q$, the set $\{\I_q\}_{q\in\Q}$ is a \emph{box partition of $\Domain$} (or simply a \emph{box partition}), if each $\I_q\subseteq \Domain$ is a box, $\cup_{q\in\Q}\I_q=\Domain$, and $\I_q\cap \I_{q'}=\emptyset$ for all $q,q'\in\Q$. For $q\in\Q$, let $\ul{\bx}_q=\{\ul{x}_{q,l}\}_{l\in\Links}$, $\ol{\bx}_q=\{\ol{x}_{q,l}\}_{l\in\Links}$ denote the lower  and upper corners, respectively, of $\I_q$, that is, $\I_q=\{\bx\mid \ul{\bx}_q\prec_q^1\bx\prec_q^2\ol{\bx}_q\}$ where $\prec^1_q=\{\prec^1_{q,l}\}_{l\in\Links}$, $\prec^2_q=\{\prec^2_{q,l}\}_{l\in\Links}$, and $\prec^1_{q,l}, \prec^2_{q,l}\in\{<,\leq\}$.
\end{definition}

For a box partition $\{\I_q\}_{q\in\Q}$ of $\Domain$, let $\pi:\Domain\to\Q$ be uniquely defined by the condition $x\in\I_{\pi(x)}$, that is, $\pi(\cdot)$ is the natural projection from the domain $\Domain$ to the (index set of) boxes.
A special case of a box partition of a rectangular domain is the following:
\begin{definition}[Gridded box partition]
  For $\Domain=\{\bx=\{x_l\}_{l\in\Links}\mid \ux_l\leq x_l\leq\barx_l\}$, a box partition $\{\I_q\}_{q\in\Q}$ of $\Domain$ is a \emph{gridded} box partition if for each $l\in\Links$, there exists $N_l\in\{1,2,\ldots\}$ and a set of intervals $\{I^l_1,\ldots,I^l_{N_l}\}$ such that $\cup_{i=1}^{N_l}I^l_i=[\ux_l,\barx_l]$ and for each $q\in\Q$, there exists indices $q_l\in\{1,\ldots,N_l\}$ such that $\I_q=\prod_{l\in\Links}I^l_{q_l}$.  For gridded box partitions, we make the identification $\Q\cong\prod_{l\in\Links}\{1,\ldots,N_l\}$ for all $l\in\Links$.
\end{definition}

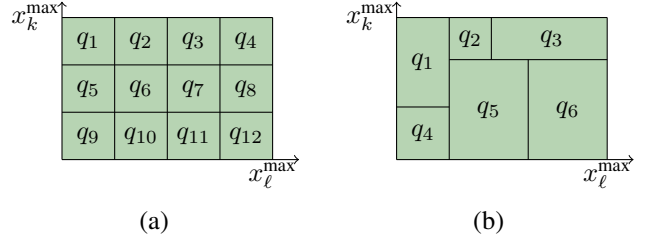
\begin{figure}
  \centering
\begin{tabular}{c c c}
  \begin{tikzpicture}[scale=.7]
      \def\c{1};
\fill[dgreen,opacity=.3] (0,0) rectangle(4,2.7);

\draw[->] (0,0) -- (4.5,0);
\draw[->] (0,0) -- (0,3);

\draw[-] (1,0) -- +(0,2.7);
\draw[-] (2,0) -- +(0,2.7);
\draw[-] (3,0) -- +(0,2.7);
\draw[-] (4,0) -- +(0,2.7);
\draw[-] (0,.9) -- +(4,0);
\draw[-] (0,1.8) -- +(4,0);
\draw[-] (0,2.7) -- +(4,0);
\node at (.5,.45) {$q_9$};
\node at (1.5,.45) {$q_{10}$};
\node at (2.5,.45) {$q_{11}$};
\node at (3.5,.45) {$q_{12}$};
\node at (.5,1.35) {$q_5$};
\node at (1.5,1.35) {$q_6$};
\node at (2.5,1.35) {$q_7$};
\node at (3.5,1.35) {$q_8$};
\node at (.5,2.25) {$q_1$};
\node at (1.5,2.25) {$q_2$};
\node at (2.5,2.25) {$q_3$};
\node at (3.5,2.25) {$q_4$};

\node at (4,-.3) {$x_l^\text{max}$};
\node at (-.5,2.7) {$x_k^\text{max}$};
\end{tikzpicture}&
   \begin{tikzpicture}[scale=.7]
\fill[dgreen,opacity=.3] (0,0) rectangle(4,2.7);

\draw[->] (0,0) -- (4.5,0);
\draw[->] (0,0) -- (0,3);

\draw[-] (1,0) -- (1,2.7);
\draw[-] (0,1) -- (1,1);
\draw[-] (1,1.9) -- (4,1.9);
\draw[-] (1.8,1.9) -- (1.8,2.7);
\draw[-] (2.5,0) -- (2.5,1.9);
\draw[-] (0,2.7) -- +(4,0);
\draw[-] (4,0) -- +(0,2.7);
\node at (4,-.3) {$x_l^\text{max}$};
\node at (-.5,2.7) {$x_k^\text{max}$};
\node at (.5,1.85) {$q_1$};
\node at (.5,.5) {$q_4$};
\node at (1.75,.95) {$q_5$};
\node at (1.4,2.3) {$q_2$};
\node at (2.95,2.3) {$q_3$};
\node at (3.25,.95) {$q_6$};
  \end{tikzpicture}\\
(a)&(b)
\end{tabular}
\caption{Stylized depictions of two box partitions. (a) A gridded box partition with regularly sized intervals. (b) A nongridded box partition.}
  \label{fig:box}
\end{figure}

When a box partition is not a gridded box partition, we say it is \emph{nongridded}.
Fig. \ref{fig:box} shows two examples of box partitions, one of which is a gridded box partition.
From a box partition of the traffic network domain $\Domain$, we obtain a finite state representation, or \emph{abstraction}, of the traffic network model as follows. Each element of the box partition corresponds to a single state in the resulting finite state transition system, and to obtain a computationally tractable approach, we propose a method for efficiently obtaining a finite state abstraction using the componentwise monotonicity properties developed above:

\begin{definition}[CM-induced finite state abstraction]
\label{def:aqts}
 Given a box partition $\{\I_q\}_{q\in\Q}$ of $\Domain$, the nondeterministic \emph{componentwise monotonicity-induced (CM-induced) finite state abstraction}, or simply the \emph{finite state abstraction},  of the traffic model is the transition system $\T=(\Q,\Sig,\to)$ where $\Q$ is the index set of the box partition, $\Sig$ is the available signaling inputs, and $\to$ is defined by:
  \begin{align}
    \label{eq:31}
    (q,\bsig,q')\in\to\quad \text{ if and only if }\quad  \I_{q'}\cap\oPost(\cl(\I_q),\bsig)\neq \emptyset.
\end{align}
\end{definition}

\begin{rem}
  We must take the closure of $\I_q$ in \eqref{eq:31} as the $\oPost$ operator and relevant properties (\emph{e.g.}, \eqref{eq:30}) assume a closed box. This allows efficient algorithms for constructing $\to$ via \eqref{eq:31} as detailed below.
\end{rem}
Note that the CM-induced finite state abstraction is nondeterministic. Nondeterminism arises from the disturbance input $\bd$ and from the fact that a collection of continuous states is abstracted to one discrete state.

By the definition of the finite state abstraction above, for any trajectory $\bx[t]$, $t\in\mathbb{N}$ generated by the traffic model under input sequence $\bs[t]$, $t\in\mathbb{N}$, there exists a unique sequence $q[t]$, $t\in\mathbb{N}$ with each $q[t]\in\Q$ such that $x[t]\in \I_{q[t]}$ and $q[t]\overset{\bs[t]}{\to}q[t+1]$. A transition system satisfying this property is said to be a \emph{discrete abstraction} of the dynamical system~\eqref{eq:11}. A controller synthesized from the abstraction to satisfy an LTL formula as described in Section \ref{sec:cw-monot-traff} can be applied to the original traffic network with the same guarantees because the abstraction \emph{simulates} the original traffic network \cite{Baier:2008vn}. However, abstractions generally result in unavoidable conservatism, that is, nonexistence of an appropriate control strategy from the abstraction does not imply nonexistence of a control strategy for the original traffic network.

The following corollary to Remark \ref{rem:tight} implies that the finite state abstraction suggested in Definition \ref{def:aqts} does not introduce excessive conservatism; specifically, Corollary \ref{cor:2} tells us that if $(q,\bsig,q')\in\to$, then for each link $l$, it is possible for the state of link $l$ to transition from a state in box $\I_q$ to a state in $\I_{q'}$. %

\begin{corollary}
\label{cor:2}
For the CM-induced finite state abstraction defined above, $(q,\bsig,q')\in\to$ if and only if
\begin{align}
  \label{eq:35}
\quad &\exists \bd=\{d_l\}_{l\in\Links}\in\Dist, \exists \bx'=\{x'_l\}_{l\in\Links}\in\I_{q'}\text{ such that } \\
\label{eq:35-2}&\forall l\in\Links,\exists \bx\in\cl(\I_q) \text{ s.t. } x'_l=F_l(\xloc_l[t],\sloc[t],d_l[t]).
\end{align}
\end{corollary}
\begin{proof}
  (if). Suppose \eqref{eq:35}--\eqref{eq:35-2} holds for some $q,q'\in \Q$ and $\bsig\in\Sig$, and let $\bd\in\Dist$ and $\bx'\in\I_{q'}$ be a particular solution such that \eqref{eq:35-2} holds for all $l$. We will show that $\bx'\in\oPost(\cl(\I_q),\bsig)$. Let $i^*$ be such that $\bd\in\Dist_{i^*}$, and let $\ul{\bxi}^l$, $\ol{\bxi}^l$ be as in Corollary \ref{cor:1}. We must have
  \begin{align}
    \label{eq:36}
    F_l(\ul{\bxi}^l,\sloc_l,\ul{d}^{i^*}_l)\leq x'_l\leq     F_l(\ol{\bxi}^l,\sloc_l,\ol{d}^{i^*}_l)
  \end{align}
by Lemma \ref{lem:reach} where we make the same substitutions as in the proof of Corollary \ref{cor:1} because \eqref{eq:22} holds for $\bx$ satisfying \eqref{eq:35-2} for each $l\in\Links$. By \eqref{eq:30}, it follows that $\bx'\in\oPost(\cl(\I_q),\bsig)$, and thus $(q,\bsig,q')\in\to$.

(only if). Suppose $(q,\bsig,q')\in\to$, it follows that $\I_{q'}\cap\oPost(\cl(\I_q),\bsig)\neq \emptyset$, let $\bx'\in \I_{q'}\cap\oPost(\cl(\I_q),\bsig)$ and let $i^*\in\{1,\ldots,n_\Dist\}$ be such that $ F_l(\ul{\bxi}^l,\sloc_l,\ul{d}^{i^*}_l)\leq x'_l \leq F_l(\ol{\bxi}^l,\sloc_l,\ol{d}^{i^*}_l)$ for all $l\in\Links$. Remark~\ref{rem:tight} implies that for each $l$, there exists $\bx\in\cl(\I_q)$ and $d^\dagger_l\in[\ul{d}^{i^*}_l,\ol{d}^{i^*}_l]$ such that $x'_l=F_l(\xloc_l,\sloc_l,d^\dagger_l)$. Indeed, suppose not, then
\begin{align}
  \label{eq:34}
\tilde{x}_l&\triangleq \sup_{x\in\cl(\I_q),d_l\in[\ul{d}^{i^*}_l,\ol{d}^{i^*}_l]}F_l(\xloc_l,\sloc_l,d_l)<x_l',\quad \text{or}\\
  \label{eq:34-2}\undertilde{x}\! \hspace{2pt} _l&\triangleq \inf_{x\in\cl(\I_q),d_l\in[\ul{d}^{i^*}_l,\ol{d}^{i^*}_l]}F_l(\xloc_l,\sloc_l,d_l)>x_l'.
\end{align}
If \eqref{eq:34} holds, {then $F_l(\bx,\bsig,\bd)\leq \tilde{x}_l<x_l'\leq F_l(\ol{\bxi}^l,\sloc_l,\ol{d}^{i^*}_l)$ for all $\ul{\bx}_q\leq \bx\leq \ol{\bx}_q$ and all $\ul{\bd}^{i^*}\leq \bd \leq \ol{\bd}^{i^*}$, which implies the upper bound in \eqref{eq:22} is not achieved, contradicting the first statement of Remark \ref{rem:tight}.}
A symmetric argument shows that if \eqref{eq:34-2} holds, then Remark \ref{rem:tight} is again contradicted. Defining $\bd=\{d^\dagger_l\}_{l\in\Links}$ for the particular collection $\{d^\dagger_l\}_{l\in\Links}$ above implies that \eqref{eq:35}--\eqref{eq:35-2} holds, completing the proof.
\end{proof}
{We remark that, in \eqref{eq:35-2}, the same choice of $x\in\cl(\I_q)$ will generally not work for all $l\in\Links$ due to the over-approximation of the reachable set; see \cite{Coogan:2014bh} for further discussion.}

\subsection{Constructing The Transition System $\T$}
\label{sec:constructing-t}

\begin{figure}
  \centering
  \algnewcommand\algorithmicinput{\textbf{Input:}}
\algnewcommand\INPUT{\item[\algorithmicinput]}
\algnewcommand\algorithmicinitial{\textbf{initialize:}}
\algnewcommand\INIT{\item[\algorithmicinitial]}
\algnewcommand\algorithmicoutput{\textbf{Output:}}
\algnewcommand\OUTPUT{\item[\algorithmicoutput]}
\begin{minipage}{\linewidth}
\begin{algorithmic}[1]
\algblockdefx{InputS}{EndInputS}{\textbf{inputs: }}{}
\algtext*{EndInputS}
\Function{Abstraction}{{\it network model}, $\Dist$, $\{\I_q\}_{q\in\Q}$} \textbf{returns} $\T$
\InputS {\it network model}, a traffic network model with
\Statex \pushcode update functions $\{F_l\}_{l\in\Links}$ with domain $\Domain$ 
\Statex \pushcode and signal input set $\Sig$
\State \pushcodeb $\Dist$, the disturbance set $\Dist=\cup_{i=1}^{n_\Dist} \Dist^i$
\State \pushcodeb $\{\I_q\}_{q\in\Q}$, a box partition $\Domain$
\EndInputS 
\State $\to:=\emptyset$
\For{\textbf{each } $\bsig\in\Sig$}
\For{\textbf{each }$q\in\Q$}
\For{$i:=1$ \textbf{to} $n_\Dist$}
\State $\ul{\bxi}^l:=$ as in \eqref{eq:28}
\State $\ol{\bxi}^l:=$ as in \eqref{eq:28-2}
\State $\ul{\by}:= F_l(\ul{\bxi}^l,\sloc_l,\ul{d}^i_l)$
\State $\ol{\by}:=F_l(\ol{\bxi}^l,\sloc_l,\ol{d}^i_l)$
\State $\Q':=\text{\sc{Successors}}({\ul{\by}},{\ol{\by}},\{\I_q\}_{q\in\Q})$
\State $\to:=\to\cup (q\times\bsig\times\Q')$
\EndFor
\EndFor
\EndFor
\State \Return $\T:=(\Q,\Sig,\to)$%
\EndFunction
\end{algorithmic}
\end{minipage}
  \caption{Algorithm for computing a finite state abstraction of the traffic dynamics. The algorithm requires function {\sc Successors}, which can be implemented using different algorithms, depending on the structure of the box partition.}
\label{fig:algo}
\end{figure}

We begin with the primary algorithm for calculating $\T$ shown in Fig. \ref{fig:algo}, which relies on Corollary \ref{cor:1} to compute $\oPost$ and to construct the finite state abstraction as defined in Definition \ref{def:aqts}. This algorithm requires a function called \textsc{Successors} that takes the lower and upper corners of a box $Y$ as input, as well as a box partition of $\Domain$, and returns the indices of the box partitions which intersects $Y$. We first present a generic algorithm for \textsc{Successors} applicable to any box partition. To this end, consider the nonempty box $\I_q=\{\bx\mid \ul{\bx}\prec^1_q\bx\prec^2_q\ol{\bx}\}$ and let $Y\triangleq \{\bx\mid \ul{\by}\leq \bx \leq \ol{\by}\}$. It is straightforward to show that $\I_q\cap Y\neq \emptyset$ if and only if $\ul{\bx}\prec^1_q \ol{\by}\text{ and }\ul{\by}\prec^2_{q} \ol{\bx}$.

The algorithm in Fig. \ref{fig:univ} utilizes this fact to compute $\Q'$, the indices of the partitions that intersect a box defined by the corners $\ul{\by}$ and $\ol{\by}$. The algorithm is convenient because it works for any box partition of $\Domain$, however it requires comparing the corners $\ul{\by}$, $\ol{\by}$ to the corners of each box $\I_q$, $q\in\Q$. Thus, computing $\T$ scales quadratically with $|\Q|$ since we must determine if $\oPost(\bsig,\I_q)$ intersects each box $\I_{q'}$, $q'\in\Q$ for each $q\in\Q$.

However, the general algorithm in Fig. \ref{fig:univ} fails to take into account any structure in the partition itself. For example, for gridded box partitions, we can identify $\Q'$ by comparing the corners $\ul{\by}$, $\ol{\by}$ componentwise to the partition's constituent coordinate intervals. For simplicity of presentation, we consider gridded box partitions $\{\I_q\}_{q\in\Q}$  where, for each $l\in\Links$, there exists a set of intervals $\{I_1^l,\ldots,I^l_{N_l}\}$ of the form
\begin{align}
  \label{eq:43}
  I_1^l=[\eta^l_0,\eta^l_1], \quad I_j^l=(\eta^l_{j-1},\eta^l_{j}], \ j=2,\ldots,N_l
\end{align}
for $0=\eta^l_0\leq \eta^l_1 <\eta^l_2<\ldots<\eta^l_{N_l-1}<\eta^l_{N_l}=\denscap_l$ such that $\I_q=\prod_{l\in\Links} I_{q_l}^l$ for $q=\{q_l\}_{l\in\Links}\in \Q\cong\prod_{l\in\Links}\{1,\ldots,N_l\}$. 
Define
\begin{align}
  \label{eq:44-2} \ol{j}_l&=\begin{cases}
1&\text{if }\ol{y}_l=0\\
\displaystyle \max_{j\in\{1,\ldots,N_l\}} j \text{ s.t. }\eta^l_{j-1}< \ol{y}_l&\text{else}
\end{cases}\\
  \label{eq:44} 
\ul{j}_l&= \min_{j\in\{1,\ldots,N_l\}} j \text{ s.t. }\ul{y}_l\leq \eta^l_j
\end{align}
and let $\Q'=\{\{q_l\}_{l\in\Links}\mid q_l\in\{\ul{j}_l,\ul{j}_l+1,\ldots,\ol{j}_l\}\}$. Then $\I_q\cap Y\neq \emptyset$ if and only if $q'\in\Q'$.
Thus, to determine the partitions $\Q'$ that intersect a given box $Y$, we simply identify the indices of the intervals that intersects $Y$ along each dimension. Finding $\ul{j}_l$ and $\ol{j}_l$ can be done in $O(N_l)$ time for each $l$, thus solving for $\Q'$ requires $O(|\Links|\max_{l\in\Links} \{N_l\})$ time. Thus, for gridded box partitions, we can instead use the implementation of \textsc{Successors} found in Fig. \ref{fig:euc}.

\begin{figure}
  \centering
  \begin{minipage}{1.0\linewidth}
    \begin{algorithmic}[1]
      \algblockdefx{InputS}{EndInputS}{\textbf{inputs: }}{}
      \algtext*{EndInputS}
      \Function{Successors}{$\ul{\by}$, $\ol{\by}$, $\{\I_q\}_{q\in\Q}$} \textbf{returns} $\Q'$
      \InputS $\ul{\by}$ and $\ol{\by}$, points in domain $\mathcal{X}$
      \State \pushcodeb $\{\I_q\}_{q\in\Q}$, an interval partition of $\Domain$
      \EndInputS
      \State \textbf{initialize: }$\Q'=\emptyset$
      \For{\textbf{each }$q'\in\Q$}
      \If{ ($\ul{\bx}_q \prec^1_q \ol{\by}$)$\land$($\ul{\by}\prec^2_q \ol{\bx}_q$)}
      \State $\Q':=\Q'\cup\{q'\}$
      \EndIf
      \EndFor
      \State \textbf{return} $\Q'$
      \EndFunction
    \end{algorithmic}
  \end{minipage}
  \caption{A generic algorithm for overapproximating successor states applicable to any box partition. The algorithm returns $\Q'$, the set of indices of boxes that intersect the box defined by the corners $\ul{\by}$, $\ol{\by}$, that is, $q'\in\Q'$ if and only if $\I_{q'}\cap \{\bx\in\Domain\mid \ul{\by}\leq \bx \leq \ol{\by}\}\neq \emptyset$.}
\label{fig:univ}
\end{figure}

\begin{figure}
  \centering
  \begin{minipage}{1.0\linewidth}
    \begin{algorithmic}[1]
      \algblockdefx{InputS}{EndInputS}{\textbf{inputs: }}{}
      \algblockdefx[NAME]{IfS}{EndIfS} [2]{\textbf{if} #1 \textbf{then} #2 \textbf{else}}{}
      \algtext*{EndInputS}
      \algtext*{EndIfS}
      \Function{Successors}{$\ul{\by}$, $\ol{\by}$, $\{\I_q\}_{q\in Q}$} \textbf{returns} $\Q'$
      \InputS $\ul{\by}=\{\ul{y}_l\}_{l\in\Links}$ and $\ol{\by}=\{\ol{y}_l\}_{l\in\Links}$,
      \Statex \pushcode points in domain $\mathcal{X}$
      \State \pushcodeb $\Q$, a grid interval partition of $\Domain$
      \EndInputS
      \For{\textbf{each} $l\in\Links$}
      \State $\ol{j}_l:=$ as in \eqref{eq:44-2}
      \State $\ul{j}_l:=$ as in \eqref{eq:44}
      \EndFor
      \State \textbf{return }$\Q':=\left\{(j_l)_{l\in\Links}\mid j_l\in\{\ub{j}_l,\ldots,\bar{j}_l\}\ \forall l\in\Links\right\}$
      \EndFunction
    \end{algorithmic}
  \end{minipage}
  \caption{An algorithm for identifying successor states when $\Q$ is a gridded box partition.}
\label{fig:euc}
\end{figure}

The algorithm in Fig. \ref{fig:euc} may be applied to nongridded box partitions with some modification. In particular, a nongridded box partition $\{\I_q\}_{q\in Q}$ can be \emph{refined} to obtain the coarsest possible gridded box partition with the property that each box $\I_q$ is the union of boxes from the refinement. This refinement is used as an index set; to compute the possible transitions from $\I_q$ for $q\in\Q$ under signaling $\bsig\in\Sig$, we compute $\ul{\by}$ and $\ol{\by}$ as in lines 11 and 12 of the algorithm in Fig. \ref{fig:algo}, and then use the refinement along with the algorithm in Fig. \ref{fig:euc} to determine $\Q'$, the set of intersected boxes. The refinement does not introduce additional states in the transition system or require addition reach computations; it is only used to efficiently determine $\Q'$. For example, the coarsest refinement of Fig. \ref{fig:box}(b) partitions the box labeled $q_5$ into four boxes, which are all labeled $q_5$.
This method will be faster if the total number of intervals in the refinement is less than $|\Q|$.

\subsection{Augmenting the State Space with Signaling}
To capture control objectives that include the state of the signals themselves (which are modeled as inputs in the finite state abstraction $\T$), we augment the discrete state space. Examples of specifications that require this augmention include $\ltlphi_2$ and $\ltlphi_3$ above or the specifications ``the state of an intersection cannot change more than once per $n^\text{min}$ time steps'' or ``an input signal cannot remain unchanged for $n^\text{max}$ time steps.'' In particular, we propose augmenting the finite state abstraction to encompass both the current state of the finite state abstraction and the current state of the traffic signals.

\begin{definition}[Augmented finite state abstraction] 
\label{def:aug}
The \emph{augmented finite state abstraction} of the traffic network is the transition system $\Taug=(\QQ,\Sig,\toaug)$ where

  \begin{itemize}[leftmargin=10pt]
  \item $\QQ=\Q\times \Sig$ is the set of discrete states consisting of the box partition index set and the set of allowed input signals,
\item $\Sig$ is the set of allowed input signals,
\item $\toaug\subseteq \QQ\times\Sig\times \QQ$ is the set of transitions given by $((q,\bsigma),\bsig,(q',\bsigma'))\in\toaug$ for $(q,\bsigma), (q',\bsigma')\in\QQ$ if and only if $(q,\bsig,q')\in\to$ and $\bsigma'=\bsig$.
  \end{itemize}

\end{definition}

\section{Synthesizing Controllers from LTL Specifications}
\label{sec:synth-contr-from}

\subsection{Synthesis Summary}
\label{sec:synthesis-summary}
We omit the details of how a control strategy is synthesized from the nondeterministic transition system $\Taug$ for a given LTL control objective, as this is well-documented in the literature, see \emph{e.g.} \cite{Yordanov:2012fk, horn2005streett}. Instead, we summarize the main steps of this synthesis as follows: from the LTL control objective, we obtain a deterministic Rabin automaton that accepts all and only trajectories that satisfy the LTL specification using off-the-shelf software. %
 We then construct the synchronous product of the Rabin automaton and $\Taug$ in Definition \ref{def:aug}, resulting in a nondeterministic product Rabin automaton from which a control strategy is found by solving a Rabin game \cite{horn2005streett}. The result is a control strategy for which trajectories of the traffic network are guaranteed to satisfy the LTL specification.

As the discrete state space is finite, the signaling control strategy takes the form of a collection of ``lookup'' tables over the discrete states of the system, $\QQ$, and there is one such table for each state in the Rabin automaton. Thus, implementing the control strategy requires implementing the underlying deterministic transition system of the specification Rabin automaton, which is interpreted as a finite memory controller that ``tracks'' progress of the LTL specification and updates at each time step.  Given the current state of the Rabin transition system, the controller chooses the signaling input dictated by the current state of the augmented system $\QQ$. {Thus, we obtain a state feedback, finite memory controller. Additionally, the controller update only requires knowledge of the currently occupied {partition} of $\Q$, and thus does not require precise knowledge of the state $\bx$.}

\subsection{Computational Requirements}
\label{sec:comp-requ}

For each $q\in\Q$ and each $\bsig\in\Sig$, determining the set $\{q'\mid q\overset{\bsig}{\to}q'\}$ requires first computing $\oPost(\I_q,\bsig)$, which requires computing $F_l(\cdot)$ at $2n_\Dist$ points for each $l\in\Links$. Since $F_l(\cdot,\bsig,\cdot)$ is only a function of the links in $\Lloc$, each computation of this function requires time $O(1)$ assuming the average number of links at {an intersection} does not change with network size. Thus $\oPost(\I_q,\bsig)$ is computed in time $O(|\Links|n_\Dist)$. Then, we identify the set $\Q'$ of boxes that intersect $\oPost(\I_q,\bsig)$. As described in Section \ref{sec:constructing-t}, this requires $2|\Q|$ comparisons of vectors of length $|\Links|$ and thus is done in time $O(|\Q||\Links|)$ via the algorithm in Fig. \ref{fig:univ}. However, for gridded box partitions, $\Q'$ is computed in time $O(|\Links|\max_{l\in\Links}\{N_l\})$ by the algorithm in Fig. \ref{fig:euc}. Even for nongridded box partitions, $\Q'$ can be computed in time $O(|\Links|\max_{l\in\Links}\{N_l\})$ where $N_l$ is interpreted as the number of intervals of link $l$ resulting from the coarsest refinement of the box partition that results in a gridded box partition.  For a gridded partition, $|\Q|=\prod_{l\in\Links}N_l$ and thus the number of boxes grows exponentially with the number of links in the network. For a nongridded box partition, the number of partitions can be substantially lower. %
Since $\{q'\mid q\overset{\bsig}{\to}q'\}$ must be computed for each $q$ and $\bsig$, constructing $\T$ requires time $O(|\Q|^2|\Sig||\Links|^2n_\Dist)$ when using the algorithm in Fig. \ref{fig:univ} or time $O(|\Q||\Sig|\max_{l\in\Links}\{N_l\}|\Links|^2n_\Dist)$ for the algorithm in Fig.~\ref{fig:euc}. 

We briefly compare these computational requirements to that of polyhedral methods such as those in \cite{Yordanov:2012fk}. As the dynamics in \eqref{eq:5}--\eqref{eq:7-2} are piecewise affine, such methods can in principle be applied here. Computing $\Post(\I_q,\bsig)$ requires polyhedral affine transformations and polyhedral geometric sums, operations that scale exponentially in $|\Links|$ \cite{Kurzhanskiy:2010hc, Herceg:2013db}. To determine if $\Post(\I_q,\bsig)$ intersects another polytope, geometric differences are required, which again scales exponentially with $|\Links|$. %

\section{Case Study}
\label{sec:case-study}
\begin{figure}[t]
  \centering
  \begin{tikzpicture}
    [subsystem/.style={quad with diagonal fill, ns color=red, ew color=green, text =black, draw,inner sep=1pt,minimum size=4mm},
     nostate/.style={gray!50!white},
>=to,line width=1pt, >=latex, scale=1]
\node[subsystem,label={[label distance=-.2cm]225:$v_1$}] (node1) at (0,0) {};
\node[subsystem,label={[label distance=-.2cm]225:$v_2$}] (node2) at (1.5,0) {};
\node[subsystem,label={[label distance=-.2cm]225:$v_3$}] (node3) at (3,0) {};
\node[subsystem,label={[label distance=-.2cm]225:$v_4$}] (node4) at (4.5,0) {};
\node (A) at (-1.5,0) {};
\node (B) at (6,0) {};
\draw[->, line width=2pt] (A) -- node[above,pos=.4]{$1$}(node1);
\draw[->, line width=2pt] (node1) -- node[above,pos=.4]{$2$} (node2);
\draw[->, line width=2pt] (node2) -- node[above,pos=.4]{$3$} (node3);
\draw[->, line width=2pt] (node3) -- node[above,pos=.35]{$4$} (node4);
\draw[->, nostate, line width=2pt] (node4) -- (B);
\draw[<-] (node1.110) -- node[left,pos=.6]{$6$} +(0,.5);
\draw[->,nostate] (node1.70) -- +(0,.5);
\draw[<-] (node1.290) -- node[right,pos=.6]{$5$} +(0,-.5);
\draw[->,nostate] (node1.250) -- +(0,-.5);
\draw[<-] (node4.110) -- node[left,pos=.6]{$10$} +(0,.5);
\draw[->,nostate] (node4.70) -- +(0,.5);
\draw[<-] (node4.290) -- node[right,pos=.6]{$9$} +(0,-.5);
\draw[->,nostate] (node4.250) -- +(0,-.5);
\draw[<-] (node2.90) -- node[left,pos=.6]{$7$} +(0,.5);
\draw[->,nostate] (node2.270) -- +(0,-.5);
\draw[->, nostate] (node3.90) -- +(0,.5);
\draw[<-] (node3.270) -- node[right,pos=.6]{$8$} +(0,-.5);
\end{tikzpicture}
  \caption{Signalized network consisting of a major corridor road (links 1, 2, 3, and 4) which intersects minor cross streets (links 5, 6, 7, 8, 9, and 10). The gray links are not explicitly modeled.}
  \label{fig:ex2}
\end{figure}
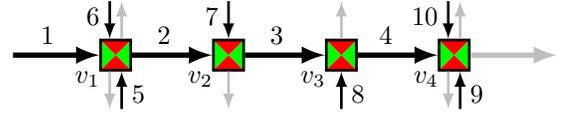

 We consider the example network in Fig. \ref{fig:ex2} which consists of a main corridor (links 1, 2, 3, and 4) with intersecting cross streets (links 5, 6, 7, 8, 9, and 10) and four intersections, a commonly encountered network configuration. The gray links exit the network and are not explicitly modeled. The network parameters are
$(\denscap_1,\ldots,\denscap_{10})=(40,50,50, 50, 40, 40, 40,40,40,40)$,
$(c_1,\ldots,c_{10})=(20,20,20, 20, 10, 10, 10,10,10,10)$,
$\beta_{12}=\beta_{23}=\beta_{34}=\beta_{62}=\beta_{52}=0.5$, 
$\beta_{73}=\beta_{84}=0.9$, $\alpha^{\{1\}}_{62}=\alpha^{\{1\}}_{52}=0.5$, and all other supply ratios are one, where the time step is 15 seconds. We assume
 \begin{align}
\notag   \Dist=&\{\bd\mid \mathbf{0}\leq \bd \leq [10\ 0 \ 0 \ 0 \ 10 \ 10 \ 0 \ 0 \ 10 \ 10]\}\\
   \label{eq:45}&\cup  \{\bd\mid \mathbf{0}\leq \bd \leq [10\ 0 \ 0 \ 0 \ 10 \ 10 \ 10 \ 10 \ 0 \ 0]\}.
 \end{align}
We further assume the available signals are $\Sig_{v_1}=\{\{1\},\{5, 6\}\}$, $\Sig_{v_2}=\{\{2\},\{7\}\}$, $\Sig_{v_3}=\{\{3\},\{8\}\}$, and $\Sig_{v_4}=\{\{4\},\{9, 10\}\}$. 
We wish to find a control policy for the four signalized intersections that satisfies the LTL property $\ltlphi=\ltlphi_1\And \ltlphi_2 \And \ltlphi_3\And \ltlphi_4$
where
\begin{flalign}
\notag\ltlphi_1=&\Box\Diam(\bsig_{v_1}=\{5, 6\}) \And \Box\Diam(\bsig_{v_2}=\{7\})\\
  \label{eq:46}&\hspace{.5in} \And\Box\Diam(\bsig_{v_3}=\{8\}) \And \Box\Diam(\bsig_{v_4}=\{9,10\})\hspace{-.1in}\\
\notag&\hspace{-.1in}\text{``Each signal actuates cross street traffic infinitely often''}\\
\label{eq:46-2}\ltlphi_2=&\Diam \Box \big((x_1\leq 30) \And (x_2\leq 30) \And (x_3\leq 30) \And (x_4\leq 30)\big)\hspace{-10pt}\\
\notag&\text{``Eventually, links 1, 2, 3, and 4 have fewer than 30}\\[-.3em]\notag&\text{vehicles on each link and this remains true for all time''}\\
\label{eq:46-3}\ltlphi_3=&\Box\big(\neg(\bsig_{v_4}=\{4\})\And \Next (\bsig_{v_4}=\{4\}) \Impl \Next\Next  (\bsig_{v_4}=\{4\}) \big)\hspace{-10pt}\\
\notag \ltlphi_4=&\Box\big(\neg(\bsig_{v_4}=\{9,10\})\And \Next (\bsig_{v_4}=\{9,10\})\\
  \label{eq:46-4}&\hspace{1.3in}\Impl \Next\Next  (\bsig_{v_4}=\{9,10\}) \big)\hspace{-3pt}\\
\notag&\text{For $\ltlphi_3$ (resp. $\ltlphi_4$), ``The signal at {intersection} $v_4$ must }\\[-.3em]\notag&\text{actuate corridor traffic (resp. cross street traffic) for at}\\[-.3em]\notag&\text{least two sequential time-steps.''}
\end{flalign}

\begin{figure}
  \centering
\begin{tabular}{c c}
  \includegraphics[width=2.55in]{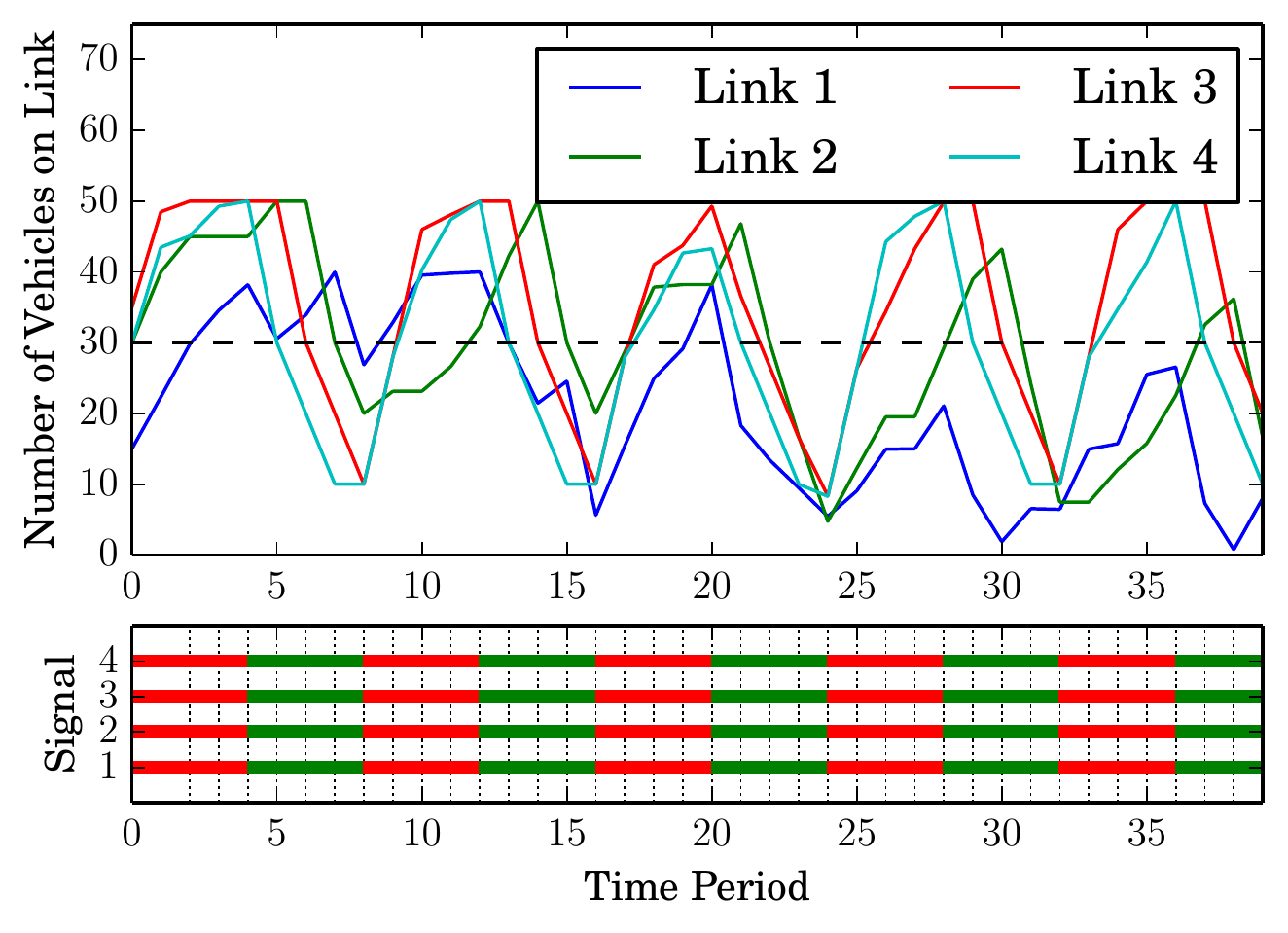}\\
(a)\\
  \includegraphics[width=2.55in]{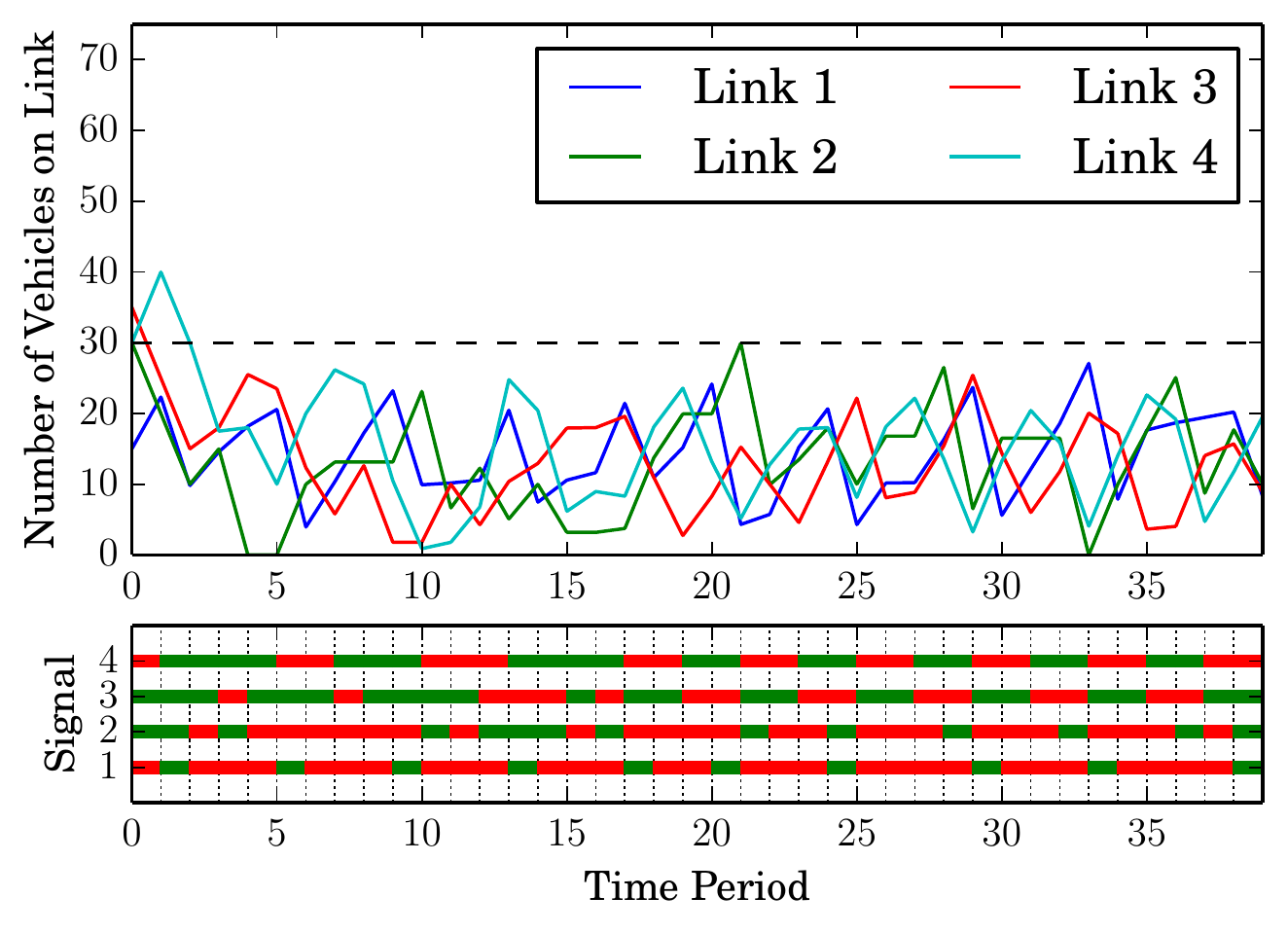}\\
(b)
\end{tabular}
\caption{(a) A sample trajectory of a na\"{i}ve strategy that alternately actuates corridor traffic and then cross street traffic for four time steps each in a synchronized fashion. This policy does not satisfy the desired control objective, in particular, \eqref{eq:46-2} is not satisfied. (b) A sample trajectory resulting from the synthesized control policy that is guaranteed to satisfy the LTL policy \eqref{eq:46}--\eqref{eq:46-4}. In the lower plots of (a) and (b), green (resp., red) for the signal trace indicates corridor traffic (resp., cross street traffic) is actuated.}
  \label{fig:trace}
\end{figure}

Thus, $\ltlphi_2$ reflects our preference for actuating corridor traffic and ensures that eventually, links 2, 3, and 4 have ``adequate supply'' because if the number of vehicles on these links is less than 30, then these links can always accept upstream demand, thus avoiding \emph{congestion} (congestion occurs when demand is greater than supply). Condition $\ltlphi_1$ ensures that, despite the preference for facilitating traffic along the corridor, we must infinitely often actuate traffic at the cross streets. Conditions $\ltlphi_3$ and $\ltlphi_4$ are needed if, \emph{e.g.} there exists crosswalks at {intersection} $v_4$ and a minimum amount of time is required to allow pedestrians to cross. {Note that repeated application of the $\Next$ (``next'') operator allows us to consider finite time horizons as in \eqref{eq:46-3} and \eqref{eq:46-4}.}

We partition the state space into 408 boxes that favors larger boxes when there are fewer total vehicles in the network. There are 16 signaling inputs, and thus, the number of states in the transition system $\Taug$ is $|\QQ|=6528$. The Rabin automaton generated from $\ltlphi$ contains 62 states and one acceptance pair. Computing the finite state abstraction $\T$ took 22.4 seconds. In contrast, the computation would be intractable using polyhedral methods. Computing the product automaton took 30.9 minutes and computing the control strategy took 15.5 minutes on a Macbook Pro with a 2.3 GHz processor where we use the Rabin game solver in \texttt{conPAS2} \cite{Yordanov:2012fk}, however \texttt{conPAS2} is written in MATLAB and the synthesis process is likely to be much more efficient if implemented in C or C++ and optimized. Furthermore, all computations can be performed offline and some are parallelizable, such as computing the product automaton. Finally, we note that the computed control strategy is implemented with minimal online costs.

 Fig. \ref{fig:trace}(a) shows a sample trajectory of the network using a na\"{i}ve coordinated signaling strategy whereby each intersection actuates corridor traffic for three time steps and then cross traffic for three time steps. The exogenous disturbance is generated uniformly randomly from $\Dist$. The trajectories are not guaranteed to satisfy the control objective, in particular, $\ltlphi_2$ is violated. Fig. \ref{fig:trace}(b) shows a sample trajectory of the system with a control strategy synthesized using the finite state abstraction augmented with signal history and the LTL requirement above.   The control strategy is correct-by-construction and thus guaranteed to satisfy $\ltlphi$ from any initial state.

We see that the synthesized controller reacts to increased vehicles on the corridor by actuating the corridor links, thereby preventing congestion (inadequate supply) along the corridor. At the same time, the controller actuates cross streets when doing so does not adversely affect conditions on the corridor (\emph{i.e.}, cause congestion). In contrast, the fixed time controller in Fig \ref{fig:trace}(a) is not able to react to the current conditions of the network and fails to prevent congestion along the corridor; in fact, links 2, 3, and 4 periodically reach full capacity.

\section{Conclusions}
\label{sec:conclusions}
We have proposed a framework for synthesizing a control strategy for a traffic network that ensures the resulting traffic dynamics satisfy a control objective expressed {in linear temporal logic (LTL)}. %
In addition to offering a novel domain for applying formal methods tools in a control theory setting, we have identified and exploited key properties of traffic networks to allow efficient computation of a finite state abstraction.%

Future research will investigate systematic methods for determining an appropriate box partition to further reduce the number of states in the computed abstraction. Additionally, traffic networks are often composed of tightly coupled neighborhoods and towns connected by sparse longer roads, and such networks may be amenable to a compositional formal methods approach using an assume-guarantee framework \cite{clarke1999model}.%

\bibliographystyle{ieeetr}
\bibliography{newbib}

\begin{IEEEbiography}[{\includegraphics[width=1in,height=1.25in,clip,keepaspectratio]{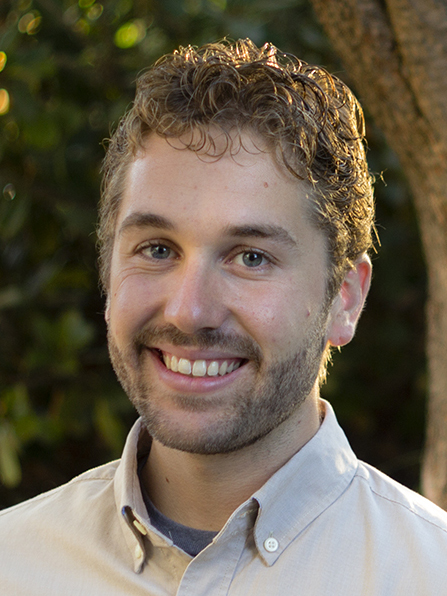}}]{Samuel Coogan}
is a Ph.D. candidate in Electrical Engineering and Computer Sciences at the University of California, Berkeley. He received his B.S. in Electrical Engineering from Georgia Tech in 2010 and his M.S. in Electrical Engineering from UC Berkeley in 2012. His research interests are in control theory, nonlinear and hybrid systems, and formal methods. He is particularly interested in applying techniques from these domains to the control and design of transportation systems. He received an NSF Graduate Research Fellowship in 2010 and the Leon O. Chua Award for outstanding achievement in nonlinear science from UC Berkeley in 2014.
\end{IEEEbiography}

\begin{IEEEbiography}[{\includegraphics[width=1in,height=1.25in,clip,keepaspectratio]{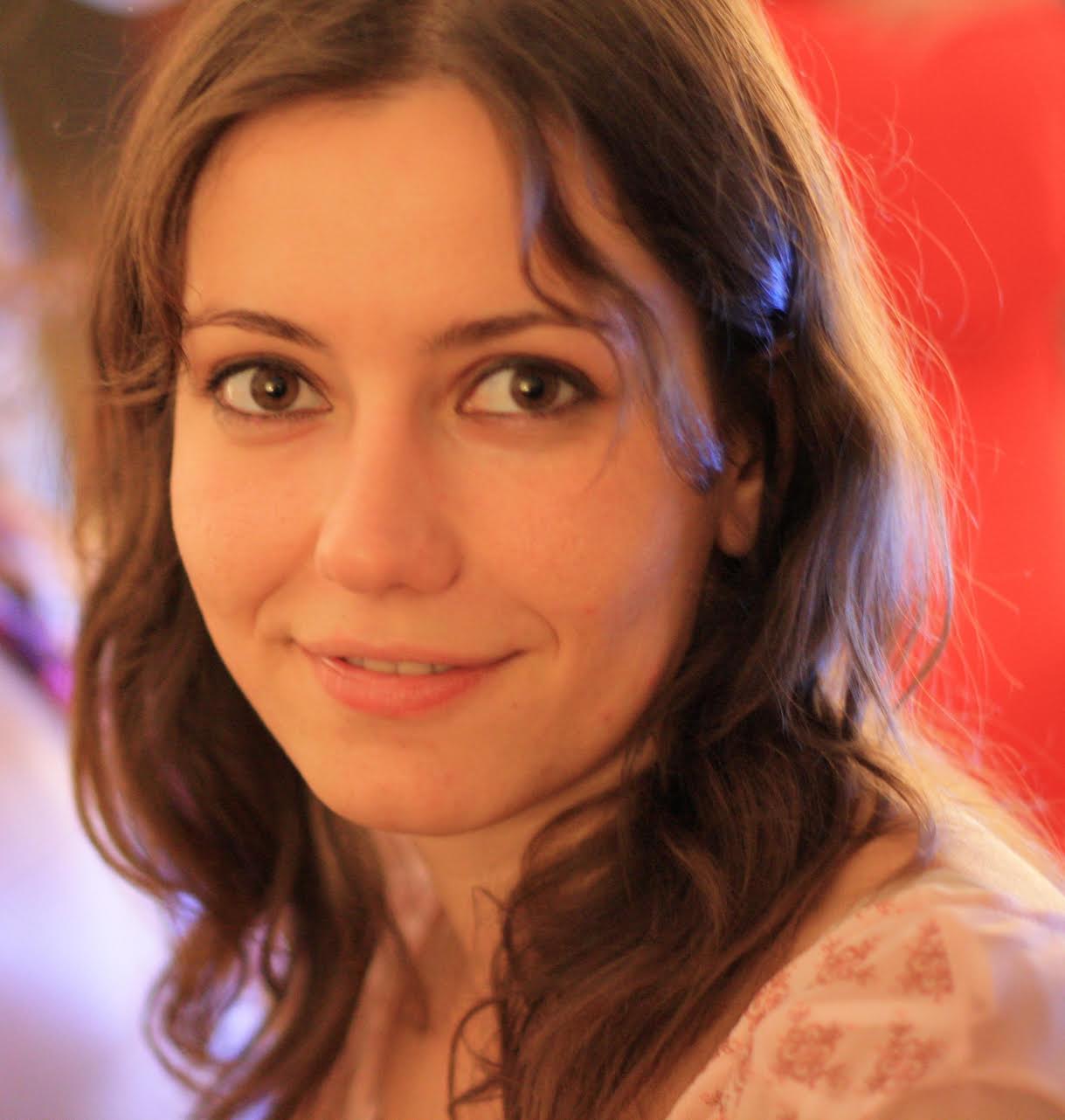}}]{Ebru Aydin Gol}
received her B.Sc. degree in computer engineering from Orta Dogu Teknik Universitesi, Ankara, Turkey, in 2008, M.Sc. degree in computer science from Ecole Polytechnique Federale de Lausanne, Lausanne, Switzerland, in 2010 and Ph.D. degree in systems engineering from Boston University, Boston, MA, USA in 2014. She has been a Site Reliability Engineer at Google since 2014. Her research interests include verification and control of dynamical systems, optimal control, and synthetic biology.
\end{IEEEbiography}

\begin{IEEEbiography}[{\includegraphics[width=1in,height=1.25in,clip,keepaspectratio]{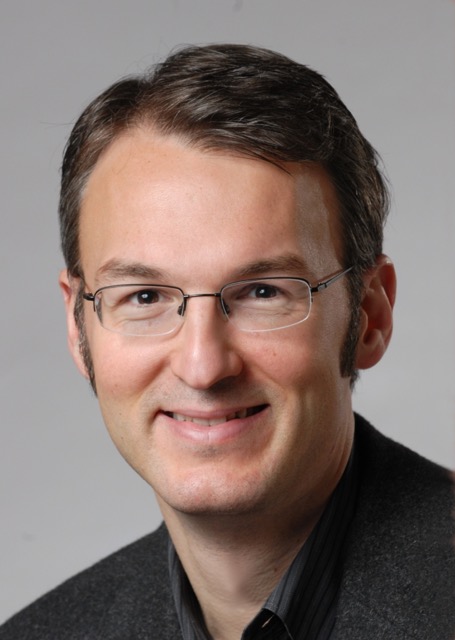}}]{Murat Arcak}
is a professor at U.C. Berkeley in the Electrical Engineering and Computer Sciences Department.  He received the B.S. degree in Electrical Engineering from the Bogazici University, Istanbul, Turkey (1996) and the M.S. and Ph.D. degrees from the University of California, Santa Barbara (1997 and 2000). His research is in dynamical systems and control theory with applications to synthetic biology, multi-agent systems, and transportation. Prior to joining Berkeley in 2008, he was a faculty member at the Rensselaer Polytechnic Institute. He received a CAREER Award from the National Science Foundation in 2003, the Donald P. Eckman Award from the American Automatic Control Council in 2006, the Control and Systems Theory Prize from the Society for Industrial and Applied Mathematics (SIAM) in 2007, and the Antonio Ruberti Young Researcher Prize from the IEEE Control Systems Society in 2014. He is a member of SIAM and a fellow of IEEE.
\end{IEEEbiography}

\begin{IEEEbiography}[{\includegraphics[width=1in,height=1.25in,clip,keepaspectratio]{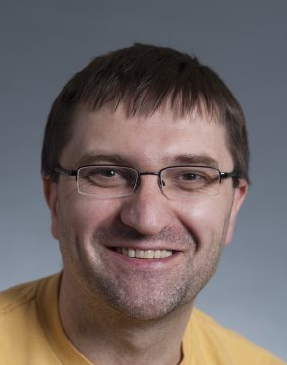}}]{Calin Belta}
is a Professor in the Department of Mechanical Engineering, Department of Electrical and Computer Engineering, and the Division of Systems Engineering at Boston University, where he is also affiliated with the Center for Information and Systems Engineering (CISE) and the Bioinformatics Program. His research focuses on dynamics and control theory, with particular emphasis on hybrid and cyber-physical systems, formal synthesis and verification, and applications in robotics and systems biology. Calin Belta is a Senior Member of the IEEE and an Associate Editor for the SIAM Journal on Control and Optimization (SICON) and the IEEE Transactions on Automatic Control. He received the Air Force Office of Scientific Research Young Investigator Award and the National Science Foundation CAREER Award. 
\end{IEEEbiography}

\end{document}